\documentclass[runningheads,a4paper,orivec,envcountsame,envcountsect,final]{llncs}
\usepackage[UKenglish]{babel}
\usepackage[babel]{csquotes}
\usepackage{amssymb,amsmath,amsxtra,mathrsfs,mathtools,bm,stmaryrd}
\usepackage{booktabs}
\usepackage{multirow}
\usepackage{nicefrac}
\usepackage{colonequals}
\usepackage{xspace}
\usepackage{dsfont}
\usepackage{tikz}
\usetikzlibrary{arrows,arrows.meta,automata,calc,decorations.markings,decorations.pathreplacing,decorations.pathmorphing,fit,math,positioning,shapes,shapes.geometric,shapes.callouts,shapes.misc}
\usepackage{ifthen}
\usepackage[notref,notcite]{showkeys}
\usepackage{seqsplit}
\usepackage{xstring}
\usepackage[noadjust]{cite}
\usepackage{rotating}
\usepackage{tikz-cd}
\usepackage{blkarray}
\usepackage{adjustbox}
\usepackage{calc}
\usepackage{dutchcal}
\usepackage{etoolbox}
\usepackage{relsize}
\usepackage{microtype}
\usepackage[footnote,marginclue,nomargin]{fixme}

\tikzset{every state/.style={inner sep=2pt, minimum size=0pt}}
\tikzset{every loop/.append style={looseness=7, out=135,in=45}}
\tikzcdset{scale cd/.style={every label/.append style={scale=#1},
    cells={nodes={scale=#1}}}}

\allowdisplaybreaks

\spnewtheorem{assumption}[theorem]{Assumption}{\bfseries}{\rmfamily}
\spnewtheorem{fact}[theorem]{Fact}{\bfseries}{\rmfamily}
\spnewtheorem{notation}[theorem]{Notation}{\bfseries}{\rmfamily}
\spnewtheorem{observation}[theorem]{Observation}{\bfseries}{\rmfamily}
\spnewtheorem{defn}[theorem]{Definition}{\bfseries}{\rmfamily}
\spnewtheorem{prop}[theorem]{Proposition}{\bfseries}{\itshape}
\spnewtheorem{lem}[theorem]{Lemma}{\bfseries}{\itshape}
\spnewtheorem{expl}[theorem]{Example}{\bfseries}{\rmfamily}
\spnewtheorem{algo}[theorem]{Algorithm}{\bfseries}{\rmfamily}
\spnewtheorem{rem}[theorem]{Remark}{\bfseries}{\rmfamily}
\spnewtheorem{construction}[theorem]{Construction}{\bfseries}{\rmfamily}
\spnewtheorem{examples}[theorem]{Examples}{\bfseries}{\rmfamily}
\spnewtheorem{example_}[theorem]{Example}{\bfseries}{\rmfamily}
\spnewtheorem*{proofsketch}{Proof Sketch}{\itshape}{\rmfamily}



\ExplSyntaxOn
\NewDocumentCommand{\makecycle}{om}{
	\ensuremath{ \left(~\guest_print_list:nn { #2 } {~{\ }~}~\right)\IfValueT{#1}{^{#1}}}
}
\NewDocumentCommand{\maketuple}{om}{%
\ensuremath{\left(~\guest_print_list:nn { #2 } {~{,\:}~}~\right)\IfValueT{#1}{^{#1}}}
}
\NewDocumentCommand{\makemonad}{om}{
	\ensuremath{ \left\langle~\guest_print_list:nn { #2 } {~{,\ }~}~\right\rangle\IfValueT{#1}{^{#1}}}
}
\NewDocumentCommand{\makegrammar}{om}{
	\ensuremath{ ~\guest_print_list:nn { #2 } {~{\ \,\vert\ \,}~}~\IfValueT{#1}{^{#1}}}
}

\seq_new:N \l_guest_list_seq
\cs_new_protected:Nn \guest_print_list:nn
{
	\seq_set_from_clist:Nn \l_guest_list_seq { #1 }
	\seq_use:Nn \l_guest_list_seq { #2 }
}
\ExplSyntaxOff

\usepackage{hyperref}
\hypersetup{hidelinks,final,bookmarks}

\newcommand{\defaultshowkeysformat}[1]{%
\StrSubstitute{#1}{ }{\textvisiblespace}[\TEMP]%
\parbox[t]{\marginparwidth}{\raggedright\normalfont\small\ttfamily\(\{\){\color{red!50!black}\expandafter\seqsplit\expandafter{\TEMP}}\(\}\)}%
}

\renewcommand*\showkeyslabelformat[1]{%
\noexpandarg%
\defaultshowkeysformat{#1}%
}

\newcommand\restr[2]{{
  \left.\kern-\nulldelimiterspace 
  #1 
  \littletaller 
  \right|_{#2} 
  }}

\DeclareMathOperator{\dom}{\textsf{dom}}
\DeclareMathOperator{\supp}{\mathsf{supp}}

\DeclareMathOperator{\alphaequiv}{\equiv_{\alpha}}

\newcommand{\littletaller}{\mathchoice{\vphantom{\big|}}{}{}{}}

\newcommand{\seq}{\subseteq}
\newcommand{\qes}{\supseteq}

\newcommand{\xra}[1]{\xrightarrow{~#1~}}

\newcommand{\xto}[1]{\xra{#1}}

\newcommand{\deriv}[2]{{#1}^{-1}#2}
\newcommand{\Der}{\mathsf{Der}}
\DeclareMathOperator{\GF}{\mathsf{GF}}
\DeclareMathOperator{\LF}{\mathsf{LF}}

\newcommand{\At}{\mathds{A}}

\newcommand{\Ats}{\At^{\!\raisebox{1pt}{\scriptsize$\star$}}}
\newcommand{\Atw}{\At^{\omega}}
\newcommand{\barAs}{\barNames^{\star}}
\newcommand{\barAw}{\barNames^{\omega}}

\newcommand{\op}[1]{\operatorname{\mathsf{#1}}}

\newcommand{\cl}{\op{cl}}

\newcommand{\card}[1]{\left\vert#1\right\vert}
\let\abs\card

\newcommand{\nf}{\ensuremath{\mathsf{nf}}}

\newcommand{\A}{\mathcal{A}}
\renewcommand{\H}{\mathcal{H}}

\newcommand{\tree}{\mathcal{T}}

\newcommand{\teach}{\textsf{T}\xspace}
\newcommand{\learn}{\textsf{L}\xspace}
\newcommand{\tass}{\textsf{TA}\xspace}

\newcommand{\bartree}[1][\barNamess]{\mathcal{T}_{#1}}

\newcommand{\cat}[1]{\mathscr{#1}}
\def\A{\cat A}

\renewcommand{\epsilon}{\varepsilon}

\newcommand{\Perm}{\mathsf{Perm}}
\newcommand{\names}{\At}
\newcommand{\Abstr}[2][\names]{[#1]#2}

\newcommand{\parnom}[2][]{\names_{#1}^{\$\mathbf{#2}}}

\newcommand{\braket}[1]{\langle #1 \rangle}
\newcommand{\fs}{\mathsf{fs}}
\newcommand{\pow}{\mathcal{P}}

\newcommand{\powfs}{\pow_{\fs}}

\newcommand{\Nat}{\mathds{N}}

\renewcommand{\phi}{\varphi}

\newcommand{\runtree}{\textsf{run}}

\newcommand{\set}[1]{\{#1\}}
\newcommand{\setw}[2]{\{#1\,:\,#2\}}

\newcommand{\Lstar}{\ensuremath{\mathsf{L}^{\ast}}\xspace}
\newcommand{\Lsharp}{\ensuremath{\mathsf{L}^{\#}}}
\newcommand{\NLstar}{\ensuremath{\mathsf{NL}^{\ast}}}

\newcommand{\nomNLstar}{\ensuremath{\nu\mathsf{NL}^{\ast}}}

\newcommand{\Jir}{\mathsf{JI}}

\newcommand{\takeout}[1]{\empty}

\tikzset{shiftarr/.style={
        rounded corners,%
        to path={--([#1]\tikztostart.center)
                     -- ([#1]\tikztotarget.center) \tikztonodes
                     -- (\tikztotarget)},
}}

\tikzset{shiftarrr/.style={
        rounded corners,%
        to path={-- ([#1]\tikztostart.center)
                    |- (\tikztotarget)  \tikztonodes},
}}

\tikzset{roundcornerarr/.style={
        rounded corners,%
        to path={--([#1]\tikztostart.south)
                     |- (\tikztotarget) \tikztonodes},
}}

\tikzset{roundcornerarrr/.style={
        rounded corners,%
        to path={ -| (\tikztotarget) \tikztonodes},
}}

\tikzset{roundcornerarrrr/.style={
        rounded corners,%
        to path={ |- (\tikztotarget) \tikztonodes},
}}

\NewDocumentCommand{\thickrarrow}{O{-1.25mm} O{1.2cm}}{%
    \begin{tikzpicture}[scale=0.5, baseline=#1, thick]
        \protect\draw[-latex, ultra thick] (0,0) -- (#2,0);
    \end{tikzpicture}%
}

\NewDocumentCommand{\dashedrarrow}{O{-1.25mm} O{1.2cm}}{%
    \begin{tikzpicture}[scale=0.5, baseline=#1, thick]
        \protect\draw[-latex, very thick, dashed] (0,0) -- (#2,0);
    \end{tikzpicture}%
}

\newcommand{\overbar}[1]{\mkern 1.5mu\overline{\mkern-1.5mu#1\mkern-1.5mu}\mkern 1.5mu}

\newcommand{\mybar}[3]{%
  \mathrlap{\hspace{#2}\overline{\scalebox{#1}[1]{\phantom{\ensuremath{#3}}}}}\ensuremath{#3}
}

\newcommand{\FN}{\mathsf{FN}}

\newcommand{\barA}{{\mybar{0.6}{2.5pt}{A}}}
\newcommand{\barAut}{{\mybar{0.4}{6pt}{\A}}} 

\newcommand{\barNames}{{\mybar{0.7}{1.25pt}{\names}}}

\newcommand{\barNamess}{{\overbar{\names}}}

\newcommand{\ub}{{\mathsf{ub}}}

\newcommand{\N}{\mathds{N}}

\newcommand{\midmid}{\hspace{0.2ex}{\rule[-0.1ex]{0.6pt}{1.65ex}}\hspace{0.2ex}}
\newcommand{\scriptmidmid}{\hspace{0.2ex}{\rule[-0.1ex]{0.6pt}{1.1ex}}\hspace{0.2ex}}

\newcommand{\newletter}[1]{{\midmid}#1}
\newcommand{\scriptnew}[1]{{\scriptmidmid}#1}

\newcommand{\newtreeletter}[1]{{\nu}#1}

\newcommand{\quotient}[2]{{#1}/{#2}} 

\let\originalleft\left
\let\originalright\right
\renewcommand{\left}{\mathopen{}\mathclose\bgroup\originalleft}
\renewcommand{\right}{\aftergroup\egroup\originalright}

\usepackage[nameinlink,capitalize,noabbrev]{cleveref}

\crefname{rem}{Remark}{Remarks}
\crefname{prop}{Proposition}{Propositions}
\crefname{defn}{Definition}{Definitions}
\crefname{expl}{Example}{Examples}
\crefname{lem}{Lemma}{Lemmas}

\usepackage{wrapfig}
\makeatletter
\renewcommand\section{\@startsection{section}{1}{\z@}%
  {-18\p@ \@plus -4\p@ \@minus -4\p@}%
  {12\p@ \@plus 4\p@ \@minus 4\p@}%
  {\normalfont\large\bfseries\boldmath
    \rightskip=\z@ \@plus 8em\pretolerance=10000 }}
\renewcommand\subsection{\@startsection{subsection}{2}{\z@}%
 {-6\p@ \@plus -2\p@ \@minus -2\p@}%
 {-0.5em \@plus -0.22em \@minus -0.1em}%
 {\normalfont\normalsize\bfseries\boldmath}}                      
\newcommand\mypar{\@startsection{paragraph}{4}{\z@}%
  {-6\p@ \@plus -4\p@ \@minus -4\p@}%
  {-0.5em \@plus -0.22em \@minus -0.1em}%
  {\normalfont\normalsize\bfseries}}
\makeatother

\usepackage[inline]{enumitem}
\setlist[enumerate,1]{label=(\arabic*),font=\normalfont,align=left,leftmargin=0pt,labelindent=0pt,listparindent=\parindent,labelwidth=0pt,itemindent=!,topsep=3pt,parsep=0pt,itemsep=3pt,start=1}
\setlist[enumerate,2]{label=(\alph*),font=\normalfont,labelindent=*,leftmargin=*,start=1}
\setlist[itemize]{labelindent=*,leftmargin=*,topsep=5pt,itemsep=3pt}
\setlist[description]{labelindent=*,leftmargin=*,itemindent=-1 em,topsep=2pt}

\begin{document}
    \title{Learning Automata with Name Allocation}
\author{Florian Frank\inst{1}\fnmsep\thanks{Supported as part of the DFG Research and Training Group 2475 \enquote{Cybercrime and Forensic Computing} (grant number 393541319/GRK2475/2-2024).}, Stefan Milius\inst{1}\fnmsep\thanks{Supported by Deutsche Forschungsgemeinschaft (DFG, German Research Foundation) -- project number 517924115.}, Jurriaan Rot\inst{2}\fnmsep\thanks{Partially supported by the NWO grant No.~VI.Vidi.223.096.}, and Henning Urbat\inst{1}\fnmsep\thanks{Supported by Deutsche Forschungsgemeinschaft (DFG, German Research Foundation) -- project number 569130867.}} 

\authorrunning{F.~Frank, S.~Milius, J.~Rot, H.~Urbat}
\institute{Friedrich-Alexander-Universität Erlangen-Nürnberg
   \email{}
  \and Radboud Universiteit Nijmegen
  \email{}}
    \maketitle
    \begin{abstract}
      Automata over infinite alphabets have emerged as a convenient computational model for processing structures involving data, such as nonces in cryptographic protocols or data values in XML documents. We introduce active learning methods for \emph{bar automata}, a species of automata that process finite data words represented as bar strings, which are words with explicit name binding letters. Bar automata have pleasant algorithmic properties. We develop a framework in which \emph{every} learning algorithm for standard deterministic or non-deterministic finite automata over finite alphabets can be used to learn bar automata, with a query complexity determined by that of the chosen learner. The technical key to our approach is the algorithmic handling of $\alpha$-equivalence of bar strings, which allows bridging the gap between finite and infinite alphabets. The principles underlying our framework are generic and also apply to bar Büchi automata and bar tree automata, leading to the first active learning methods for data languages of infinite words and finite trees.
    \end{abstract}
    \section{Introduction}\label{sec:intro}
    Active automata learning is a family of techniques for inferring an automaton from a black-box system, by interacting with this system and making observations about its behaviour. 
    Originally introduced by Angluin~\cite{angluin87}, the celebrated~\Lstar algorithm allows to effectively learn deterministic finite automata in this way. Since her work, automata learning has been combined with model checking and conformance testing~\cite{DBLP:journals/jalc/PeledVY02}, turning it into an effective tool for bug finding. Indeed, automata learning algorithms have been used to analyze and learn models of network protocols (e.g.~\cite{DBLP:conf/ndss/Fiterau-Brostean23,FiterauEtAl17,FJV16}), legacy code~\cite{AslamCSB20,SHV16}, embedded software~\cite{SmeenkMVJ15} and interfaces of software components~\cite{HowarISBJ12}; see~\cite{DBLP:journals/cacm/Vaandrager17,HowarS16} for further references.
	
    Angluin's original \Lstar algorithm has been improved in various ways. State-of-the-art algorithms such as $\mathsf{TTT}$~\cite{ihs14} and \Lsharp~\cite{vgrw22} may substantially reduce the number of
    queries needed during learning. Orthogonally there have been numerous extensions of \Lstar-type
    algorithms to models beyond classical deterministic finite automata, including non-deterministic finite automata~\cite{bhkl09}, Mealy machines~\cite{MargariaNRS04}, quantitative automata~\cite{bm15,hkrs20},
    tree automata~\cite{dh03,k13}, automata for languages of infinite
    words~\cite{af16,MalerP95,fcctw08,lczl21,blls25}, and, most relevant for this paper, automata over infinite alphabets, namely register automata~\cite{DBLP:conf/tacas/DierlFHJST24,CasselHJS16,CEGAR12,BolligHLM13}, data automata~\cite{dhlt14}, and nominal automata~\cite{mssks17}. Infinite alphabets represent data, for
    example, nonces in cryptographic protocols~\cite{KurtzEA07}, data
    values in XML documents~\cite{NevenEA04}, object
    identities~\cite{GrigoreEA13}, or parameters of method
    calls~\cite{HowarEA19}. The principal challenge underlying all learning algorithms for data languages is to cleverly represent infinite data by finite means.

    All the above learning algorithms for data languages apply to \emph{deterministic} automata models over \emph{finite} data words. Going beyond this setting turns out to be substantially more challenging than in the case of finite alphabets. The only non-trivial learning algorithm so far for non-deterministic automata over infinite alphabets, due to Moerman and Sammartino~\cite{ms22}, applies to \emph{residual} non-deterministic register automata, an (undecidable) semantic subclass of non-deterministic register automata over finite words. In terms of expressivity, this class lies strictly between deterministic and general non-deterministic register automata. No non-trivial learning algorithm for data languages of trees or infinite words is known. For the case of infinite words, the difficulty lies in the fact that data languages accepted by non-deterministic (Büchi) register automata are not uniquely representable by their ultimately periodic words. In the setting of finite alphabets, such a representation is crucial and constitutes the basis for all existing learning algorithms for regular $\omega$-languages~\cite{af16,MalerP95,fcctw08,lczl21,blls25}.

\mypar{Contribution}  We approach the problem of actively learning non-deterministic register automata by studying \emph{bar automata}, a species of automata over infinite alphabets introduced in recent years in versions for finite data words (\emph{bar word automata})~\cite{skmw17}, infinite data words (\emph{bar Büchi automata})~\cite{uhms21} and finite data trees (\emph{bar tree automata})~\cite{ps24}. Bar automata yield a finite representation of corresponding nominal models (coalgebras over nominal sets) with explicit name allocation, namely regular non-deterministic nominal automata (RNNAs)~\cite{skmw17}, Büchi RNNAs~\cite{uhms21} and regular non-deterministic nominal tree automata (RNTAs)~\cite{uhms21}.

The key difference between bar automata and other models for data languages, such as non-deterministic register automata~\cite{KaminskiFrancez94,KaminskiZeitlin10} or the equivalent non-deterministic nominal orbit-finite automata~\cite{BojanczykEA14}, is the use of \emph{binding transitions} while at the same time restricting to {finite} rather than orbit-finite branching. In this way, bar automata retain a reasonable level of expressivity; they correspond to a natural \emph{syntactic} subclass of non-deterministic register automata (unlike residual automata) with certain lossiness conditions, and are incomparable to deterministic ones. The central feature of bar automata are their pleasant algorithmic properties: most notably, language inclusion is decidable in space
polynomial in the size of the automata and exponential in the \emph{degree} (number of registers). This is in sharp contrast to standard non-deterministic register automata where inclusion is undecidable for automata with more than two registers~\cite{KaminskiFrancez94} (or one register in terms of the definition by Demri und Lazi\'c~\cite{DemriLazic09}). 

The standard semantics of bar automata is at the level of \emph{bar languages}. For
instance, bar word automata consume \emph{bar strings}, which are finite words formed from
\emph{plain letters} $a$ and \emph{bar letters} $\newletter a$ with $a$ taken from an
infinite domain $\names$ of names (representing data values). Intuitively, a bar string is a
pattern that determines how letters are read from the input and stored in memory for future
comparison, where an occurrence of $\newletter a$ is interpreted as reading an input letter
and binding this letter to the name $a$. Accordingly, bar word languages are sets of bar
strings modulo an obvious notion of $\alpha$-equivalence. Similar principles apply to bar
languages of infinite bar strings and bar trees. Bar languages have direct uses as models of
terms in the $\lambda$/$\pi$/$\mu$-calculus; for instance, in~\cite[Ex.~3.7]{ps24} a
property of $\pi$-terms is modelled by suitable bar tree automata. More importantly, bar
languages represent data languages (i.e.~sets of words/trees over $\names$ without bar
letters) in two ways corresponding to two slightly different disciplines of
$\alpha$-renaming of bound names: \emph{global freshness} means that binding transitions
read names that have not occurred before (as in session automata~\cite{BolligEA14}), 
\emph{local freshness} means that names are not currently stored in memory (as in register
automata).

Our main contribution is an active learning method for bar automata that applies {uniformly} to bar word, bar Büchi and bar tree automata.
More specifically, we demonstrate that \emph{any} learning algorithm for classical automata over
(in)finite words or trees can be extended to a learning algorithm for the associated type of
bar automata; our extension is agnostic to the
choice of the underlying learning algorithm. The key technical idea to our approach is to reduce the problem of
learning an unknown bar language to learning a canonical representation of it over a restricted, finite
subalphabet. This representation, whose language is closed under
$\alpha$-equivalence w.r.t.~to the finite subalphabet, is itself a regular language over a
finite alphabet, making it amenable to classical learning.
In the case of
infinite words, this reduction process allows us to work with ultimately periodic strings only
at the level of finite alphabets, bypassing the above problems.

However, the reduction is far from immediate: there is a mismatch in the types of counterexamples that the learner receives. To see this, let us recall Angluin's \emph{minimally adequate teacher} (MAT)
framework, which allows the learner to pose membership queries (`is this input in the unknown language?') and equivalence queries (`is this hypothesis automaton~$\H$ correct?') to an oracle. For equivalence queries,
the oracle returns a counterexample whenever~$\H$ is incorrect. This is precisely where the difficulty lies: this counterexample may not be useful to learn
the canonical representation of the language over the finite subalphabet.

The core ingredient to processing counterexamples and resolving this mismatch is to effectively
find $\alpha$-equivalent words over that subalphabet. We therefore introduce new techniques for
checking $\alpha$-equivalence of bar strings and trees.

Our
methods for deciding $\alpha$-equivalence are non-trivial and of independent
interest, as they do not just arise in automata learning but potentially also in other
algorithms handling bar automata, such as minimization, reduction, or model checking.  By combining these
techniques for $\alpha$-equivalence with our reductions to known learning algorithms for
word/tree automata or automata over infinite words, we obtain learning algorithms for bar
languages over each of these variants. By learning bar automata and interpreting them under local or global freshness, our algorithms can be understood as learning data languages. The classes of data languages learnable in this way are incomparable to existing learning algorithms for register or nominal automata; see `Related Work' and \Cref{sec:bar-automata}. 

In summary, the main contributions of this paper are twofold:
(1) We give a reduction of the problem of learning bar automata to that of learning classical automata for regular languages of finite words, infinite words, or finite trees; and (2) we develop an approach for effectively
checking $\alpha$-equivalence of finite and infinite bar strings as well as bar trees.

\mypar{Related Work} Learning automata models for data languages of finite words without explicit binding is an active area of research.
Bollig et al.~\cite{BolligHLM13} introduced
a learning algorithm for \emph{session automata}, which are register automata requiring some
data values to be \emph{fresh} for the whole string. They are equivalent to a subclass of bar
word automata under global freshness~\cite{skmw17}.  Their learning
algorithm
uses a somewhat similar idea to ours in that it reduces the learning problem to the \Lstar algorithm for deterministic
finite automata. Our algorithm can learn a non-deterministic, hence more succinct representation of session automata.

Cassel et al.~\cite{CasselHJS16} and Dierl et al.~\cite{DBLP:conf/tacas/DierlFHJST24} learn \enquote{determinate} register automata on finite words, which are semantically equivalent to deterministic ones and are incomparable to bar automata under local freshness.
Similarly, Aarts et al.~\cite{afkv15} learn deterministic register automata with outputs which are again incomparable to bar automata under local freshness.
Moerman et al.~\cite{mssks17,ms22} present learning algorithms for deterministic nominal automata as well as for \emph{residual} non-deterministic nominal automata on finite
words. Residual nominal automata are incomparable to bar automata under local freshness (\Cref{ex:nonresidual}).

Sakamoto~\cite{saka97} proposes a learning algorithm for a subclass of deterministic register automata on finite
words, parametric in a learning algorithm for DFAs. These automata are incomparable to bar automata under local freshness.

To our knowledge, there is no prior work on learning Büchi or tree data
languages. Our paper yields learning algorithms for subclasses of both kinds.

 \section{Active Automata Learning}\label{sec:active-learning} Automata learning is about systematically inferring an automaton for an
    \emph{unknown} language~$L_\teach$. Most approaches are based on Angluin's framework~\cite{angluin87} of a \emph{minimally adequate teacher} (MAT), a game between a \emph{learner} that aims to infer~$L_\teach$, and a \emph{teacher} that the learner can ask for pieces of information about $L_\teach$. To this end, the learner runs an algorithm generating a sequence $\H_1,\H_2,\H_3,\ldots$ of automata (\emph{hypotheses}) yielding increasingly better approximations of $L_\teach$. To improve the current hypothesis $\H_i$, the learner can direct two types of questions to the teacher: \emph{membership queries} (`is a given input in $L_\teach$?') and \emph{equivalence queries} (`is the hypothesis $\H_i$ correct, i.e.~accepting $L_\teach$?'). The teacher's reply to an equivalence query is either `Yes', in which case the learner has successfully inferred~$L_\teach$, or a \emph{counterexample}, that is, an input on which $\H_i$ and $L_\teach$ differ. Note that inputs of membership queries are chosen by the learner, making this setting \emph{active}. In contrast, in \emph{passive} learning only a fixed set of  examples of elements and non-elements of $L_\teach$ is given. Active learning has been studied for numerous automata models. We consider three cases:

\mypar{Finite Automata} A \emph{(non-)deterministic finite automaton} (\emph{DFA}/\emph{NFA}) $\A = (Q, A, \to, q_0, F)$ is given by a finite set $Q$ of states, a finite input alphabet $A$, an initial state $q_0\in Q$, a set $F\seq Q$ of final states, and a transition relation $\to\,\seq Q\times A\times Q$, which in the deterministic case forms a function from $Q\times A$ to $Q$. The \emph{language $L(\A)\seq A^*$ accepted by $\A$} is the set of all finite words $w=a_1\cdots a_n$ over $A$ for which there exists an accepting run, that is, a sequence of transitions $q_0\xto{a_1} q_1\xto{a_2}\cdots \xto{a_n} q_n$ ending in a final state. Languages accepted by finite automata are called \emph{regular}. In MAT-based learning of finite automata, one assumes an unknown regular language $L_\teach\seq A^*$ (with known alphabet $A$) and admits the learner to put two types of queries to the teacher:
    \begin{description}
        \item[Membership Queries $(\textsf{MQ})$:] Given $w \in A^*$, is $w\in L_\teach$?

        \item[Equivalence Queries $(\textsf{EQ})$:] Given a hypothesis (i.e.~a finite automaton) $\H$,
            is $L(\H)=L_\teach$? If not, then the teacher returns a word in $L(\H)\oplus L_\teach$. 
    \end{description}
(Here $X\oplus Y=X\setminus Y\cup Y\setminus X$ is the \emph{symmetric difference} of sets $X$ and $Y$.)
   The classical learning algorithm for DFAs is Angluin's \Lstar~\cite{angluin87}, which learns the minimal DFA for $L_\teach$ with a number of queries polynomial in the number of states of that DFA and the maximum length of the counterexamples provided by the teacher. Several improvements of \Lstar have been proposed~\cite{rs93,kv94,ihs14,vgrw22}, based on clever processing of counterexamples and/or enhanced data structures for representing the information gained so far. Additionally, \Lstar has been adapted by Bollig et al.~\cite{bhkl09} to a learning algorithm \NLstar for (residual) NFAs.

\mypar{Büchi Automata} A \emph{Büchi automaton} is an NFA $\A = \maketuple{Q, A, \to, q_0, F}$ interpreted over infinite words. The
                \emph{language} $L(\A)\seq A^\omega$ \emph{accepted by $\A$} is given by those infinite words $w=a_1a_2a_3\cdots$ with an accepting run, that is, an infinite sequence of transitions $q_0\xto{a_1} q_1\xto{a_2} q_2\xto{a_3} \cdots$ where some final state occurs infinitely often. Languages accepted by Büchi automata are called \emph{regular $\omega$-languages}. Any such language is uniquely determined by its \emph{ultimately periodic words}, i.e.~words of the form $uv^\omega$ with $u \in A^*$ and $v \in A^+$~\cite{mcn66}. An ultimately periodic word $uv^\omega$ can be represented by the pair $(u,v)$. This representation is not unique since different pairs $(u_1,v_1)$ and $(u_2,v_2)$ may satisfy $u_1v_1^\omega=u_2v_2^\omega$. 

MAT-based learning of an unknown regular $\omega$-language $L_\teach\seq A^\omega$ involves
\begin{description}
        \item[Membership Queries $(\textsf{MQ})$:] Given $\maketuple{u, v} \in
            A^* \times A^+$, is $uv^\omega\in L_\teach$?
        \item[Equivalence Queries $(\textsf{EQ})$:] Given a hypothesis (i.e.~a Büchi automaton) $\H$,
            is $L(\H)=L_\teach$? If not, then the teacher returns a pair $\maketuple{u, v} \in A^* \times A^+$ such that $uv^\omega \in L(\H) \oplus L_\teach$.
    \end{description}
    There are active learning algorithms for regular $\omega$-languages that infer Büchi automata~\cite{fcctw08}
    or their representation via families of DFAs~\cite{lczl21,blls25}. These build on the standard \Lstar{} algorithm to learn the set of
    ultimately periodic words of $L_\teach$, from which a Büchi automaton can be derived.

\mypar{Tree Automata}  A \emph{signature}
    $\Sigma$ is a set of symbols $f,g,h,\cdots$ each with a finite \emph{arity} $n \in \Nat$. We denote an $n$-ary symbol $f$ by $\nicefrac{f}{n}$. A \emph{$\Sigma$-tree} is a finite ordered tree where every node is labeled with some symbol $f\in \Sigma$ and its number of successors is the arity of $f$. Thus, $\Sigma$-trees correspond to syntax trees, or equivalently closed terms over the signature $\Sigma$. We write $\mathcal{T}_\Sigma$ for the set of {$\Sigma$-trees}. A \emph{(non-)deterministic bottom-up finite tree automaton} (\emph{DFTA/NFTA})
    $\A = \maketuple{Q, \Sigma, \Delta, F}$ is given by a finite set $Q$ of states, a finite signature $\Sigma$, a set $F\seq Q$ of final states, and a transition relation $\Delta \seq (\coprod_{f/n\in \Sigma} Q^n)\times Q$, which in the deterministic case forms a function from the disjoint union $\coprod_{ f/n\in \Sigma} Q^n$ to~$Q$. The \emph{language}
                $L(\A)\seq \mathcal{T}_\Sigma$ \emph{accepted by $\A$} is given by those $\Sigma$-trees $t$ with an accepting run, i.e.~where the nodes of $t$ can be labeled with states from $Q$ in such a way that the labels respect transitions (if the node $f(t_1,\ldots,t_n)$ has label $q$ and $t_1,\ldots,t_n$ have labels $q_1,\ldots,q_n$ then $((f,q_1,\ldots,q_n),q)\in \Delta$) and the root is labeled with a final state. Languages accepted by finite tree automata are called \emph{regular tree languages}. Learning an unknown regular tree language $L_\teach\seq \mathcal{T}_\Sigma$ in the MAT framework involves
    \begin{description}
        \item[Membership Queries $(\textsf{MQ})$:] Given $t\in \mathcal{T}_\Sigma$, is $t\in L_\teach$?
        \item[Equivalence Queries $(\textsf{EQ})$:] Given a hypothesis (i.e.~a \emph{tree automaton}) $\H$,
            is $L(\H)=L_\teach$? If not, then the teacher returns a tree $t\in L(\H)\oplus L_\teach$.
    \end{description}
    Active learning algorithms for regular tree languages have been studied for DFTAs~\cite{dh03} and 
    (residual) NFTAs~\cite{k13}.

\section{Bar Languages and Automata}\label{sec:prelim}
We aim to extend the scope of the above learning algorithms to languages and automata over infinite alphabets, specifically to \emph{bar languages} and \emph{bar automata}. We next review several incarnations of bar languages, and their associated automata models, introduced in earlier work~\cite{skmw17,uhms21,ps24}. 

\subsection{Nominal Sets}
Bar languages are most conveniently presented within the framework of nominal sets~\cite{Pitts2013}, which offers an abstract approach to dealing with notions of name binding, $\alpha$-equivalence, and freshness. Let us recall some basic terminology from the theory of nominal sets.

    For the rest of the paper, we fix a countably infinite set $\names$ of \emph{names}, which for our purposes play the role of \emph{data values}. A \emph{finite permutation of $\names$} is a bijective map $\pi\colon \names\to\names$ such that $\pi(a) = a$ for all but finitely many $a\in \names$.
    We denote by $\Perm(\names)$ the group of all {finite permutations}, with multiplication given by composition.
    The group $\Perm(\names)$ is generated by the \emph{transpositions} $\makecycle{a, b}$ for $a \neq b \in \names$,
    where $\makecycle{a, b}$ swaps $a$ and $b$ while fixing all $c\in \names\setminus \set{a, b}$.
    A \emph{$\Perm(\names)$-set} is a set~$X$ equipped with a group action $\cdot\colon \Perm(\names) \times X \to X$,
    denoted by $(\pi,x)\mapsto \pi\cdot x$. A subset $S\seq \names$ \emph{supports} the element $x\in X$ if $\pi \cdot x = x$ for every
    $\pi\in\Perm(\names)$ such that $\pi(a) = a$ for all $a \in S$. A \emph{nominal set} is a $\Perm(\names)$-set $X$ such that every element $x\in X$ has a finite support.  
    This implies that $x$ has a least finite support, denoted by $\supp(x) \seq \names$. A name $a \in \names$ is \emph{fresh} for~$x$ if $a \notin \supp(x)$. Intuitively, we think of an element $x$ of a nominal set as some sort of syntactic object (e.g.~a string, tree, term), of $\supp(x)\seq \names$ as the (finite) set of names occurring freely in $x$, and of $\pi\cdot x$ as the result of renaming all free names in $x$ according to the permutation $\pi$ (see \Cref{ex:nominal} below). 

A subset $Y$ of a nominal set $X$ is \emph{equivariant} if $y\in Y$ implies $\pi\cdot y\in Y$ for all $\pi\in\Perm(\names)$. We write $X\times Y$ for the cartesian product of a pair $X,Y$ of nominal sets with coordinatewise action. Given a nominal set $X$ equipped with an equivariant equivalence relation $\approx\,\seq X\times X$, we write $\quotient{X}{\approx}$ for the nominal quotient set with the group action $\pi \cdot [x]_\approx = [\pi\cdot x]_\approx$.

   Finally, we need the concept of \emph{abstraction sets}, which play a vital role in the theory of nominal sets and provide semantics for binding mechanisms~\cite{gp99}.
    Given a nominal set $X$, we define the equivariant equivalence relation $\approx$ on $\names \times X$ by $\maketuple{a,x} \approx \maketuple{b,y}$ iff $\makecycle{a,c} \cdot x=\makecycle{b,c} \cdot y$ for some, or equivalently all, names $c$ that are fresh for $a$, $b$, $x$, $y$.
    The abstraction set $\Abstr{X}$ is the nominal quotient set $\quotient{(\names\times X)}{\approx}$. The $\approx$-equivalence
    class of $\maketuple{a,x} \in \names \times X$ is denoted by $\braket{a}x$.
    We may think of $\approx$ as an abstract notion of $\alpha$-equivalence and of $\braket{a}x$ as binding the name $a$ in $x$.
    Indeed, we have $\supp(\braket{a}x) = \supp(x) \setminus \set{a}$ (while $\supp\maketuple{a,x} = \set{a} \cup \supp (x)$), as expected in binding constructs.

    \begin{expl}\label{ex:nominal}
        The set $\names$ with the $\Perm(\names)$-action $\pi\cdot a = \pi(a)$ is a nominal set, as is the set
        $\Ats$ of finite words over $\names$ with the letterwise action $\pi\cdot (a_1\cdots a_n) =\pi(a_1)\cdots\pi(a_n)$.
        The least support of $a \in \names$ is the singleton set $\set{a}$, while the least support of $a_1\cdots a_n\in \Ats$ is the set $\{a_1,\ldots,a_n\}$ of its letters. 
        The abstraction set $\Abstr{\Ats}$ identifies two distinct pairs $\maketuple{a, w}$ and $\maketuple{b, v}$ iff (i) $a\neq b$, (ii) $a$ does not occur in $v$, (iii) $b$ does not occur in $w$, and (iv) $v$ emerges from $w$ by replacing $a$ with $b$ (i.e.~ $v=\makecycle{a,b}\cdot w$). For instance, $\braket{a}aa=\braket{b}{bb}$ in $\Abstr{\Ats}$, but $\braket{a}ab \neq \braket{b}bb$
        if $a \neq b$. Similarly, $\braket{b}bcb = \braket{a}aca \neq \braket{b}bbc$.
    \end{expl}

\subsection{Bar Languages}
We will work with languages of finite words, infinite words, and finite trees with binding constructs, called \emph{bar word languages}~\cite{skmw17}, \emph{bar $\omega$-languages}~\cite{uhms21}, and \emph{bar tree languages}~\cite{ps24}. Bar languages represent \emph{data languages}, i.e.~languages over the infinite alphabet $\names$ of data values.

    \mypar{Data and Bar Word Languages} A \emph{bar string} is a finite word over $\At$ (the data domain) where a bar symbol (\enquote{$\midmid$}) might precede names to indicate that the next letter is bound until the end of the word.
    Intuitively, a bar string can be seen as a pattern that determines the way letters are read from the input: an occurrence of $\newletter{a}$ corresponds to reading a letter from the input and binding this letter to the name $a$, while a free (i.e.~unbound) occurrence of $a$ means that $a$ occurs literally in the input.
    Bound names can be \emph{renamed}, giving rise to a notion of $\alpha$-equivalence of bar strings.
    The new name must be \emph{fresh}, i.e.~cannot occur freely in the scope of the binding.
    For instance, in $ba\newletter{b}ab$ the $\midmid$ binds the letter $b$ in $\newletter{b}ab$.
    The string $ba\newletter{b}ab$ therefore is $\alpha$-equivalent to $ba\newletter{c}ac$, but not to $ba\newletter{a}aa$, since $a$ occurs freely in $\newletter{b}ab$.

   These intuitions are formalized as follows. We put $\barNames = \names \cup \setw{\newletter{a}}{a \in \names}$ and refer to elements $\newletter{a}$ of $\barNames$ as \emph{bar names}, and to elements $a \in \names$ as \emph{plain names}.
        A \emph{bar string} is a finite word $w = \alpha_1 \cdots \alpha_n \in \barAs$, with \emph{length} $\abs{w} = n$.
        We turn $\barNames$ into a nominal set with the group action given by $\pi \cdot a = \pi(a)$ and $\pi \cdot \newletter{a} = \newletter{\pi(a)}$; then also $\barNames^*$ is a nominal set with group action $\pi\cdot (\alpha_1\cdots \alpha_n)=(\pi\cdot \alpha_1)\cdots (\pi\cdot \alpha_n)$.
        We define \emph{$\alpha$-equivalence} on bar strings to be the equivalence relation generated by $w\newletter{a}v \alphaequiv w\newletter{b}u$ if $\braket{a}v = \braket{b}u$ in $\Abstr{\barAs}$, and
        write $[w]_\alpha$ for the $\alpha$-equivalence class of $w$. 
        A name $a$ is \emph{free} in a bar string $w$ if there is an occurrence of the plain name $a$ in $w$ that is to the left of the first occurrence (if any) of $\newletter{a}$.
        We write $\FN(w)$ for the set of free names in $w$.
        A bar string $w$ is \emph{closed} if $\FN(w) = \emptyset$.
        It is \emph{clean} if all bar names $\newletter{a}$ in $w$ are pairwise
        distinct, and for all bar names $\newletter{a}$ in $w$ one has $a \notin \FN(w)$. 
    \begin{expl} We have $\FN(ba\newletter{b}ab)=\{a,b\}$, and $ba\newletter{b}ab\alphaequiv ba\newletter{c}ac$ for all $c\neq a$.
    \end{expl}
  Finite bar strings give rise to three different types of languages:
    \begin{defn}
        A \emph{data word language} or \emph{literal word language} is, respectively, a subset of $\Ats$ or $\barAs$.
        A \emph{bar word language} is a literal language $L\seq \barAs$ closed under $\alpha$-equivalence, that is, if $w\in L$ and $w\alphaequiv w'$ then $w'\in L$.
    \end{defn}
    \begin{rem}
     Schröder et al.~\cite{skmw17} defined bar word languages as subsets of the quotient $\quotient{\barAs\!}{\alphaequiv}$. This definition is equivalent to ours: A bar word language $L\seq \barAs$ can be identified with the subset $L'\seq \quotient{\barAs\!}{\alphaequiv}$ given by $L'= \set{[w]_\alpha \mid w\in L}$, and conversely every $L'\seq \quotient{\barAs\!}{\alphaequiv}$ yields the bar word language $L=\set{w\in \barAs \mid [w]_\alpha\in L'}$. These constructions are mutually inverse. An analogous remark also applies to bar $\omega$-languages and bar tree languages introduced below. We prefer the present definition, as it allows for simpler notation.
    \end{rem}
    Every bar word language $L \seq \barAs$ can be converted into a data word language by interpreting name binding as reading either \emph{globally fresh} letters (letters that have not been read before) or \emph{locally fresh} letters (letters not currently stored in memory).
    These two interpretations arise from two disciplines of $\alpha$-renaming as known from $\lambda$-calculus~\cite{barendregt85}, with \emph{global freshness} corresponding to a discipline of \emph{clean} renaming where bound names are never shadowed and \emph{local freshness} corresponding to an unrestricted naming discipline that allows shadowing. Formally, let $\ub(w)\in \Ats$ emerge from $w\in \barAs$ by erasing all bars; e.g.\ $\ub(ba\newletter{b}ab)=babab$. We define the data languages $\GF(L),\LF(L)\seq \Ats$ by
    \begin{equation}\label{eq:ND}
    \GF(L) = \setw{\ub(w)}{\text{$w \in L$, $w$ is clean}}\quad \text{and}\quad \LF(L) = \setw{\ub(w)}{w \in L}.
    \end{equation}
    Thus $\GF(L)$ and $\LF(L)$ yield a global and local freshness interpretation of $L$.\footnote{In the original paper~\cite{skmw17}, the operators $\GF$ and $\LF$ were called $N$ and $D$, resp.}

    The operator $\GF$ is injective (in fact, $L \seq L'$ iff $\GF(L) \seq \GF(L')$) on bar word
    languages containing only closed bar strings~\cite[Lemma~A.3]{skmw17}. Hence, such bar languages can be identified with their corresponding data languages under global freshness semantics.

    \mypar{Data and Bar $\mathbf{\omega}$-Languages}
    In addition to finite bar strings, we also consider \emph{infinite bar strings}, which are infinite words
    over $\barNames$. We let $\barAw$ denote the set of infinite bar strings; unlike $\barAs$ this is not a nominal set under the letterwise group action since infinite bar strings need not be finitely supported.  
  Free names are defined analogously to finite bar strings. Moreover, the notion of $\alpha$-equivalence extends to
     infinite bar strings as follows~\cite{uhms21}:

    \begin{defn}
        Two infinite bar strings $v, w \in \barAw$ are \emph{$\alpha$-equivalent}, denoted $v \alphaequiv w$, iff
        $v_n \alphaequiv w_n$ holds for all $n \in \Nat$, where $v_n$ and $w_n$ are the prefixes of length $n$ of $v$
        and $w$.
    \end{defn}

    Like in the case of finite bar strings, we obtain three types of languages:

    \begin{defn}
        A \emph{data $\omega$-language} or \emph{literal $\omega$-language} is, respectively, a subset of $\Atw$ or~$\barAw$.
        A \emph{bar $\omega$-language} is a literal $\omega$-language closed under $\alpha$-equivalence.
    \end{defn}

Conversions of bar $\omega$-languages into data $\omega$-languages under global and local freshness are analogous to the case of
    bar word languages; see \cite{uhms21} for details.
    \mypar{Data and Bar Tree Languages} Lastly, we consider \emph{bar tree
    languages}, recently introduced by Prucker and Schröder~\cite{ps24} (under the name \emph{alphatic tree languages}) as
    a common generalization of bar word languages and classical tree languages. We fix a finite {signature}~$\Sigma$. A \emph{bar $\Sigma$-tree} is a tree over the (infinite) signature $\barNames\times \Sigma$ that contains one symbol $\nicefrac{\alpha.f}{n}$ for every $\alpha\in\barNames$ and every $\nicefrac{f}{n}\in \Sigma$.%
\footnote{Prucker and Schröder
    write $\newtreeletter{a}.f$ instead of $\newletter{a}.f$; we prefer the latter for consistency.} We write $\bartree(\Sigma)$ for the set of {bar} {$\Sigma$-trees}, and $\bartree[S](\Sigma)\seq \bartree(\Sigma)$ for the subset of bar trees using only letters from $S\seq \barNames$. Then $\bartree(\Sigma)$ forms a nominal set with the group action defined recursively by
    $\pi \cdot (\alpha.f(t_1, \dots, t_n)) = (\pi \cdot \alpha).f(\pi\cdot t_1, \dots, \pi\cdot t_n)$.
    We obtain a notion of \emph{$\alpha$-equivalence} of bar $\Sigma$-trees by taking $\alphaequiv$ to be
    the least congruence on~$\bartree(\Sigma)$ generated by $\newletter{a}.f(t_1, \dots, t_n) \alphaequiv \newletter{b}.f(t_1',
    \dots, t_n')$ whenever $\braket{a}t_i = \braket{b}t_i'$ for $1 \leqslant i \leqslant n$.
    The set $\FN(t)$ of \emph{free names} of a bar $\Sigma$-tree $t$ is defined recursively in the expected way,
    namely by $\FN(a.f(t_1, \dots, t_n)) = \set{a} \cup \bigcup_{i = 1}^n \FN(t_i)$ and
    $\FN(\newletter{a}.f(t_1, \dots, t_n)) = \left(\bigcup_{i = 1}^n \FN(t_i)\right) \setminus \set{a}$. As in the case of finite and infinite words,  we work with three different types of languages:

    \begin{defn}
        A \emph{data tree language} or \emph{literal tree language} is, respectively, a subset of~$\bartree[\names](\Sigma)$ or~$\bartree(\Sigma)$.        
        A \emph{bar tree language} is a literal tree language closed under $\alpha$-equivalence.
    \end{defn}

   Conversions of bar tree languages into data tree languages under global and local freshness are analogous to the case of
    bar word languages; see \cite{ps24} for details.

\begin{rem}
In the following, we let the term \emph{bar language} refer to either type  of language (bar word/\mbox{$\omega$-/tree} language) introduced above. Given a bar language~$L$, we write $L\restriction \barNames_0$ for its restriction to the subalphabet $\barNames_0\seq \barNames$: \[L\restriction \barNames_0 = L\cap S \qquad\text{where}\qquad 
  S \in \set{\barNames_0^*, \barNames_0^\omega, \bartree[{\barNamess}_0](\Sigma)}.\]
\end{rem}

\subsection{Bar Automata}\label{sec:bar-automata} While data languages are commonly represented by register automata~\cite{KaminskiFrancez94}, or equivalent models over nominal sets~\cite{BojanczykEA14}, their representation by bar languages allows for an alternative and conceptually much simpler approach: every standard automata model for languages over finite alphabets can be used as a model for bar languages by restricting to finite alphabets $\barNames_0\seq \barNames$, i.e.~finite sets of bar names and plain names, and reinterpreting the usual accepted language up to $\alpha$-equivalence.
This principle has been implemented for 
classical finite automata, Büchi automata, and tree automata (\Cref{sec:active-learning}) and leads to:

    \begin{defn}\label{defn:barautomata}
        \begin{enumerate}
            \item A \emph{bar DFA}/\emph{bar NFA}~\cite{skmw17} is a DFA/NFA $\A$ with finite input alphabet $\barNames_0\seq \barNames$. The \emph{literal language} $L(\A) \seq \barNames_0^*$ \emph{accepted by $\A$} is the usual accepted language. The \emph{bar language} 
                $L_\alpha(\A) \seq \barAs$ \emph{accepted by $\A$} is the closure of $L(\A)$ under $\alpha$-equivalence, that is, $L_\alpha(\A)=\set{w\in \barAs \mid \exists w'\in L(\A).\, w\equiv_\alpha w'}$.
            \item A \emph{bar Büchi automaton}~\cite{uhms21} is a Büchi automaton $\A$ with finite input alphabet $\barNames_0\seq \barNames$. The
                \emph{literal language} $L(\A)\seq \barNames_0^\omega$ \emph{accepted by $\A$} is the usual accepted language. The \emph{bar language} $L_\alpha(\A) \seq \barAw$ \emph{accepted by $\A$} is its closure under $\alpha$-equivalence. 
            \item A \emph{bar DFTA/NFTA}~\cite{ps24} over the finite signature $\Sigma$ is a DFTA/NFTA $\A$ over the signature $\barNames_0\times \Sigma$ for some finite subset $\barNames_0\seq \barNames$. The \emph{literal language}
                $L(\A)\seq \bartree[\barNamess_0](\Sigma)$ \emph{accepted by $\A$} is the usual accepted language. The \emph{bar language} $L_\alpha(\A)\seq \bartree(\Sigma)$ \emph{accepted by $\A$} is its closure under $\alpha$-equivalence.
        \end{enumerate}
    \end{defn}

\begin{rem} 
We let the term \emph{bar automaton} refer to either automaton type above. A \emph{bar word/tree automaton} is either a bar DFA/DFTA or a bar NFA/NFTA.
\end{rem} 

    \begin{expl} \label{ex:barnfa}
      Consider the following bar NFA over  $\barNames_0 := \{\newletter{a},\newletter{b},\newletter{c},\ \newletter{d},a,b,c\}$:
      \begin{center}
          \begin{tikzpicture}[shorten >=1pt,node distance=1.5cm, auto, initial text={}]
              \node[state, initial] (q0) {$q_0$};
              \node[state] (q1) [right of=q0] {$q_1$};
              \node[state] (q2) [right of=q1] {$q_2$};
              \node[state] (q3) [right of=q2] {$q_3$};
              \node[state] (q4) [right of=q3] {$q_4$};
              \node[state] (q5) [right of=q4] {$q_5$};
              \node[state, accepting] (qf) [right of=q5] {$q_f$};
              \path[->] (q0) edge[loop above] node[midway, above] {$\newletter{a}$} ();
              \path[->] (q2) edge[loop above] node[midway, above] {$\newletter{c},\, a,\, b$} ();
              \path[->] (q4) edge[loop above] node[midway, above] {$\newletter{d},\, b,\, c$} ();
              \draw[->] (q0) -- (q1) node[midway, above] {$\newletter{a}$};
              \draw[->] (q1) -- (q2) node[midway, above] {$\newletter{b}$};
              \draw[->] (q2) -- (q3) node[midway, above] {$a$};
              \draw[->] (q3) -- (q4) node[midway, above] {$\newletter{c}$};
              \draw[->] (q4) -- (q5) node[midway, above] {$b$};
              \draw[->] (q5) -- (qf) node[midway, above] {$c$};
          \end{tikzpicture}
      \end{center}
      Its literal language $L$ is given by the regular expression 
      \[(\newletter{a})^*\newletter a \newletter b (\newletter c + a + b)^* a\newletter c(\newletter d + b + c)^*bc.\]
      Let $L_\alpha \seq \barAs$ denote its bar language. Then 
      the induced data language under local freshness is
      $\LF(L_\alpha) = \setw{{u}ab{v}ac{w}bc}{{u}, {v}, {w} \in \Ats,\ 
          a \neq b, b \neq c \in \names}$.
   There is no requirement of $a \neq c$ in the description of $\LF(L_\alpha)$ since in the regular expression
      the letter $a$ does not occur after the last $\newletter{c}$; for instance, the bar string $|a|babb|ca|cbccb|dbc \in L$
      is $\alpha$-equivalent to $|a|babb|ca|abaab|dba$, so $ababbcaabaabdba\in \LF(L_\alpha)$.
    \end{expl}
    \begin{rem}\label{rem:bar-aut-vs-nominal-aut}
      \begin{enumerate}
      \item\label{rem:bar-aut-vs-nominal-aut-1} Under bar language semantics, all of the above types of bar automata admit an
        equi-expressive model of non-deterministic nominal automata with explicit name
        allocation, namely RNNAs~\cite{skmw17}, Büchi RNNAs~\cite{uhms21} and
        RNTAs~\cite{ps24}.  Under data language semantics, bar NFAs relate to standard
        models of nominal automata and register automata (without name allocation) for data
        word languages; specifically, a bar NFA with alphabet $\barNames_0$ can be turned
        into a register automaton with the same number of states and
        $k=\card{\barNames_0\cap \names}$ registers. Under local freshness, this translation
        identifies bar NFAs with the class of \emph{name-dropping} non-deterministic
        register automata with non-deterministic reassignment~\cite{KaminskiZeitlin10}, or
        equivalently, \emph{non-guessing} and \emph{name-dropping} non-deterministic nominal
        automata~\cite{BojanczykEA14}. Under global freshness, bar NFAs are equivalent to
        \emph{session automata}~\cite{BolligHLM13}. For more details,
        see~\cite[Sec.~6]{skmw17}.
        
      \item All of the above types of bar automata correspond to coalgebras for set
        functors~\cite{rutten00}; for example, bar DFAs are coalgebras
        $Q\to \{0,1\}\times Q^{\bar{\names}_0}$ for the endofunctor
        $FQ=\{0,1\}\times Q^{\bar{\names}_0}$. The nominal models (RNNAs, Büchi
        RNNAs, RNTAs) correspond to coalgebras for endofunctors on the category of nominal
        sets~\cite{skmw17,uhms21,ps24}. Moreover, their bar language semantics is captured
        by the general framework of Kleisli-based coalgebraic trace
        semantics~\cite{sbbr13,ush16}, as shown for RNNA by Frank et al.~\cite{fmu22} and
        for Büchi RNNA by Frank~\cite{f22}. This highlights that bar language semantics is
        very natural from a coalgebraic perspective.
      \end{enumerate}
    \end{rem}

    \begin{rem}\label{rem:EquivalentModels}
      Bar DFAs and bar NFAs are equivalent under both bar language and (local or global
      freshness) data language semantics, since bar NFAs can be determinized via the power set
      construction without changing the literal language. The same applies to (bottom-up) bar
      DFTAs and bar NFTAs. This is not a contradiction to the fact that non-deterministic
      register automata are more expressive than deterministic ones~\cite{KaminskiFrancez94},
      since the translations between bar automata and register automata do not preserve
      determinism.
    \end{rem}

As a consequence of \Cref{rem:bar-aut-vs-nominal-aut}\ref{rem:bar-aut-vs-nominal-aut-1}, while our learning algorithms developed below infer bar automata, they can also be used to learn the corresponding nominal/register automata. With regard to the classes of data languages learnable in this way, our learning algorithms are incomparable to existing algorithms for deterministic and non-deterministic nominal/register automata:

    \begin{rem} \label{ex:nonresidual}
     The data language $\LF(L_\alpha) = \setw{{u}ab{v}ac{w}bc}{{u}, {v}, {w} \in \Ats,\ a \neq b, b \neq c \in \names}$ accepted by the bar NFA of~\Cref{ex:barnfa} is not \emph{residual}~\cite{ms22}, so that the algorithm $\nu$\NLstar from \emph{op.~cit.} for learning residual non-deterministic nominal automata does not apply; see appendix. In particular, this language is not accepted by any deterministic register automaton (a subclass of residual automata), hence existing learning algorithms for such automata~\cite{DBLP:conf/tacas/DierlFHJST24,CasselHJS16,CEGAR12,mssks17} also do not apply.
      On the other hand, there exist data languages accepted by deterministic nominal automata that do not allow \emph{name-dropping}~\cite{skmw17}, or residual non-deterministic nominal automata that require \emph{guessing}~\cite[Sec.~3]{ms22}, and hence are not expressible by bar NFAs. This means that our learning algorithm applies to classes of data languages orthogonal to classes captured by previous algorithms.
    \end{rem}

We conclude this section with the important observation that $\alpha$-renaming is computable at the level of bar automata.
    A bar automaton $\A$ is \emph{closed} if its literal language is closed under $\alpha$-equivalence with respect to its
    finite alphabet~$\barNames_0$, that is, 
$L(\A) = L_\alpha(\A) \restriction \barNames_0$. 

    \begin{theorem}[Closure of Bar Automata]\label{prop:barautclosure}
      For every non-deterministic bar automaton over~$\barNames_0$ with $n$ states, there
      exists a closed bar automaton over~$\barNames_0$ with $\mathcal{O}(n\cdot 2^{|\barNamess_0|\cdot (\log{|\barNamess_0|} + 1)})$ states accepting the same bar language.
    \end{theorem}
The construction of the closed automaton amounts to translating a given bar automaton to an equivalent RNNA/Büchi RNNA/RNTA and back, using the so-called \emph{name-dropping modification}~\cite{skmw17,uhms21,ps24}. Since the translation of bar automata to nominal automata is left somewhat implicit in \emph{op.~cit.} and in order to keep the presentation simple and self-contained, we give a direct construction of closures purely at the level of bar automata.

\begin{proof}[Sketch]
Let $\A = \maketuple{Q, \barNames_0,\to, q_0, F}$ be a bar NFA. Put $\names_0 := \barNames_0 \cap \names = \set{a_1, \dots, a_k}$.
      We construct the bar NFA $\barAut = \maketuple{\overline{Q}, \barNames_0, \to', \overline{q}_0, \overline{F}}$ given by
      \[\overline{Q} := Q\times \parnom[0]{k},\qquad
        \overline{q}_0 = \maketuple{q_0, (i \mapsto a_i)},\qquad
        \overline{F} := F\times \parnom[0]{k},\]
where $\parnom[0]{k}$ is the set of partial injective maps from $k$ to $\names_0$ (corresponding to partial assignments of data values from $\names_0$ to $k$ registers). The transitions are:
        \begin{itemize}
          \item $\maketuple{q, r} \xra{a}\!'{}\, \maketuple{q', r'}$ iff there is an $i \leqslant k$ such that (i)
            $r(i) = a$; (ii) $r'$ is a restriction of $r$ (i.e.~the transition can delete the contents of arbitrary registers); and (iii) $q \xra{a_i} q'$.
          \item $\maketuple{q, r} \xra{\scriptnew{a}}\!'{}\, \maketuple{q', r'}$ iff there is an
            $\alpha \in \barNames_0\setminus\names$ such that for all $i \leqslant k$, (i) if $r'(i) = a$, then $\alpha = \newletter{a_i}$; (ii) if $\alpha = \newletter{a_i}$ and $i \in \dom(r')$ (the domain of the partial map $r'$), then $r'(i) = a$; (iii) for all $i \in \dom(r')$, $r'(i) \in \set{a,r(i)}$; and (iv) $q \xra{\alpha} q'$.
        \end{itemize}
One can show that $\barAut$ literally accepts the language $L_\alpha(\A)\restriction \barNames_0$. The construction for bar Büchi automata is identical and that for bar tree automata is very similar.
\end{proof}

\begin{rem}\label{rem:tightness}
  The bound on the number of states is essentially tight: Taking $k,n\in \Nat$,
  $\barNames_0 = \setw{a_i, \newletter{a_i}}{1 \leqslant i \leqslant k}$ and a bar NFA
  accepting only the bar string $\newletter{a_1}\cdots\newletter{a_k}a_1^n\cdots a_k^n$, one can show that its closure requires a number of states
  exponential in $k$ and linear in $n$.
\end{rem}
\begin{corollary}\label{cor:restr-regular}
For every bar language $L$ accepted by some bar automaton over the finite alphabet $\barNames_0\seq\barNames$, the restriction $L\restriction \barNames_0$ is a  regular language.
\end{corollary}

\noindent
This result is the key to our approach to learning bar languages, as it enables a reduction to learning corresponding regular languages over finite alphabets. To achieve this, we need to bridge the gap between literal and bar languages, which requires the algorithmic handling of $\alpha$-equivalence.

    \section{Checking \texorpdfstring{$\boldsymbol{\alpha}$}{^^^^^^01d6fc}-Equivalence}\label{sec:checkingAlphaEquiv}

At the heart of our learning algorithms for bar automata presented in \Cref{sec:learningRNNA}, and an essential requirement for their effective implementation, is a method to check finite bar strings, ultimately periodic bar strings and bar trees for $\alpha$-equivalence. In the following we develop a suitable version of \emph{De Bruijn levels}~\cite{db72}, originally introduced as a canonical representation of $\lambda$-terms. In this way, $\alpha$-equivalence reduces to syntactic equality.
    
    \mypar{Finite Bar Strings} In the case of bar strings, the idea of De Bruijn levels is to replace all bound occurrences of some name, as well as the preceding occurrence of the bar letter creating the binding, with a suitable natural number, while leaving free names as they are. This yields a canonical representation of bar strings up to $\alpha$-equivalence.

    \begin{defn}[De Bruijn Normal Form]
        Given a bar string $w = \alpha_1 \cdots \alpha_n \in \barAs$, its \emph{De Bruijn normal form} $\nf(w) = \beta_1 \cdots \beta_n \in (\names + \Nat)^*$ is a word of the same length as $w$ over the alphabet $\names + \Nat$, with $\beta_i$ defined inductively as follows:
        \begin{enumerate}
            \item $\beta_i = \alpha_i$ if $\alpha_i\in \names$ and $\alpha_i$ is not preceded by any occurrence of $\newletter{\alpha_i}$ in $w$;
            \item $\beta_i = k$ if $\alpha_i \in \names$ and $\beta_j = k$ where $j = \max\setw{j < i}{\alpha_j = \newletter{\alpha_i}}$;
            \item $\beta_i = k + 1$ if $\alpha_i$ is a bar name and $\alpha_1 \cdots \alpha_{i - 1}$ contains $k$ bar names.
        \end{enumerate}
    \end{defn}
    \begin{expl}
The $\alpha$-equivalent bar strings $\newletter{a}c\newletter{b}b\newletter{a}a$ and $\newletter{d}c\newletter{a}a\newletter{a}a$ have the same normal form $1c2233$. The bar string $ac\newletter{a}a\newletter{b}a$ has the normal form $ac1121$.
    \end{expl}

    \begin{proposition}\label{lem:aeDBNF}
        Two bar strings $v, w \in \barAs$ are $\alpha$-equivalent iff $\nf(v) = \nf(w)$.
    \end{proposition}
    The normal form $\nf(w)$ can be computed from $w$ letter by letter in polynomial time and linear-logarithmic space
    (for representing natural numbers) in the length of $w$. Therefore, as an immediate consequence of \Cref{lem:aeDBNF}, we get:
\begin{corollary}\label{cor:alphaeq-decidable-finite-bar-strings}
 $\alpha$-equivalence of finite bar strings is polynomial-time decidable.
\end{corollary}
Indeed, to decide whether two bar strings $v,w\in \barAs$ are $\alpha$-equivalent, one computes their normal forms $\nf(v)$ and $\nf(w)$ and checks for syntactic equality.

      \mypar{Infinite Bar Strings} De Bruijn normal forms for infinite bar strings could be introduced in the same way as for finite bar strings. However, this normal form does not preserve ultimate periodicity; for instance, the periodic bar string $(\newletter aa)^\omega$ has the normal form $11 22 33\cdots$, which is not ultimately periodic. In fact, a notion of normal form that does preserve ultimate periodicity, is invariant under $\alpha$-renaming, and invariant under changing the representation of ultimately periodic bar strings via pairs of finite bar strings, seems hard to achieve.

Nonetheless, we can check ultimately periodic bar strings for $\alpha$-equivalence by reducing to the finite case. Intuitively, given $u_0v_0^\omega$ and $u_1v_1^\omega$, since $v_0$ and $v_1$ are repeated
    infinitely often, plain names in $v_i$ are either free,
    refer back to $u_i$ or to the preceeding repetition of~$v_i$. This leads to the
    following characterization:
    \begin{proposition}\label{lem:aeinfbar}
      Two ultimately periodic bar strings $u_iv_i^\omega$ $(i=0,1)$ with $\abs{u_0} = \abs{u_1}$ are
      $\alpha$-equivalent iff the prefixes $p_i = u_i v_i^{2\ell_i}$ are $\alpha$-equivalent, where
      $\ell_i = \abs{v_{1-i}}$.
    \end{proposition}
      \noindent
      The exponents $2\ell_i$ make sure that~(a) $p_0$ and $p_1$ have the same length, (b) for every pair of positions $i_0$ in $v_0$ and $i_1$ in $v_1$, there is a position $j$ in $p_0$ and $p_1$ where the letters at $i_0$ and $i_1$ `meet', and (c) that position $j$ can be chosen such that there is a full copy of $v_i$ in $p_i$ preceeding it. Note that~(a) and~(b) would already be achieved by taking just $\ell_i$ as the exponent, and the factor $2$ then achieves~(c) while maintaining~(a) and~(b).

\begin{corollary}
The $\alpha$-equivalence of ultimately periodic bar strings (represented as pairs of finite words) is decidable in polynomial time.
\end{corollary}

    \mypar{Bar Trees} For the case of bar trees, the De Bruijn
    level notation can be reduced to that for bar strings: the normal form of a bar tree is simply given by the bar string normal form
computed for every individual branch.
    \begin{defn} Let $\Sigma$ be a finite signature.
        \begin{enumerate}
            \item A \emph{De Bruijn $\Sigma$-tree} is a tree over the signature $(\At+\N)\times \Sigma$ with symbols $\nicefrac{\alpha.f}{n}$ and $\nicefrac{k.f}{n}$ for $\alpha\in\barNames$, $k\in \Nat$ and $\nicefrac{f}{n}\in \Sigma$. 
            \item Given a bar $\Sigma$-tree $t \in \bartree(\Sigma)$, its \emph{De Bruijn normal form} $\nf(t)$
                is a De Bruijn $\Sigma$-tree of the same shape as $t$. For every node $x$ of $t$ with label $\alpha.f$, the corresponding node of $\nf(t)$ has the label $\beta.f$ with $\beta\in \names+\N$ determined as follows: let $\alpha_1.f_1,\cdots,\alpha_n.f_n,\alpha.f$ be the sequence of node labels occurring on the path from the root of $t$ to $x$, and take $\beta$ to be the last letter of the bar string normal form $\nf(\alpha_1\cdots \alpha_n\alpha)\in (\names+\N)^*$.
        \end{enumerate}
    \end{defn}

    \noindent An example of a bar tree $t$ and its normal form $\nf(t)$ is shown below:
 \[
\begin{tikzpicture}[level distance=1cm, xscale=1.25, yscale=.55]
          \node (TO) at (0,0) {$a.f$}
            child { node {$\newletter{a}.g$} child { node {$a.c$} } }
            child { node {$\newletter{b}.f$} child { node {$a.g$} child { node {$b.c$} } } child { node {$\newletter{a}.f$} child { node {$b.c$} } child { node {$a.c$}} }};
          \node[anchor=east] at ($(TO.west) + (-0.5,0)$) {$\mathsf{t}:$};
        \end{tikzpicture}
\qquad
        \begin{tikzpicture}[level distance=1cm, xscale=1.25, yscale=.55]
          \node (TO) at (0,0) {$a.f$}
            child { node {$1.g$} child { node {$1.c$} } }
            child { node {$1.f$} child { node {$a.g$} child { node {$1.c$} } } child { node {$2.f$} child { node {$1.c$} } child { node {$2.c$}} }};
        \node[anchor=east] at ($(TO.west) + (-0.5,0)$) {$\nf(\mathsf{t}):$};
        \end{tikzpicture}
\]
The normal form $\nf(t)$ can be computed in
    polynomial time in the number of nodes of the given tree $t$. Like in the case of bar strings, normal forms of bar trees capture $\alpha$-equivalence and thus yield a polynomial-time decision procedure:
 
    \begin{proposition}\label{lem:aeTrees}
        Two bar trees $s, t \in \bartree(\Sigma)$ are $\alpha$-equivalent iff $\nf(s) = \nf(t)$.
    \end{proposition}
\begin{corollary}\label{cor:alphaeq-decidable-finite-bar-trees}
The $\alpha$-equivalence of bar trees is decidable in polynomial time.
\end{corollary}

    \section{Learning Bar Languages}\label{sec:learningRNNA}
With the technical preparation of the previous sections at hand, we present the main contribution of our paper: a learning algorithm for bar languages.  
\begin{notation}
In the following, we fix an \emph{unknown} bar word/$\omega$-/tree language~$L_\teach$ recognizable by some bar word/Büchi/tree automaton with finite alphabet $\barNames_0$. (In the tree case, this means a language/automaton over the signature $\barNames_0\times \Sigma$ for some finite signature $\Sigma$.)
We let the term \emph{bar input} refer to either a finite bar string/ultimately periodic bar string/bar tree.
\end{notation}
Our learning algorithm follows Angluin's MAT framework (\Cref{sec:active-learning}): it describes the strategy of a \emph{learner} $\learn_{\textsf{bar}}$ that interacts with a \emph{teacher} $\teach$ in order to infer a bar automaton over $\barNames_0$ accepting $L_\teach$.
As usual, $\learn_{\textsf{bar}}$ can direct membership and equivalence queries to \teach, which in the
present setting take the following form:
    \begin{description}
        \item[Membership Queries $(\textsf{MQ}_\alpha)$:]
          Given
          $w\in \barAs/\barAw/\bartree(\Sigma)$, is $w \in L_\teach$?
        \item[Equivalence Queries $(\textsf{EQ}_\alpha)$:]
            Given a hypothesis (i.e.\ a bar automaton) $\H$, is $L_\alpha(\H)=L_\teach$? If not,
             then
            \teach provides a counterexample, viz.\  a bar input $w_\teach\in \barAs/\barAw/\bartree(\Sigma)$ in the symmetric difference  $L_\alpha(\H)\oplus L_\teach$.
    \end{description}
 We assume that the learner $\learn_{\textsf{bar}}$ is aware of the finite alphabet $\barNames_0$ of the unknown bar automaton. Interpreting bar automata as register automata (\Cref{rem:bar-aut-vs-nominal-aut}), this corresponds to the number of registers being known. In \Cref{sec:alphabet}, we explain how to learn $L_\teach$ without this information. 

The learner $\learn_{\textsf{bar}}$ consists of two components, depicted in \Cref{fig:talearner}. 
    \begin{figure*}[t]
        \centering
        \resizebox{\linewidth}{!}{\begin{tikzpicture}[every node/.style={align=center}]
    \coordinate (TAR) at (0.8\linewidth, 0);
    \coordinate (COMPUTENODE) at (0.68\linewidth, -2.75cm);
    \coordinate (MQSHIFT) at (0,-0.5cm);
    \coordinate (MASHIFT) at (0,-0.65cm);
    \coordinate (EQSHIFT) at (0,-1.5cm);
    \coordinate (EASHIFT) at (0,-2.75cm);

    \node[anchor=north west, draw=black, thick, minimum height=4cm, minimum width=0.1\linewidth, rounded corners]
        (L) at (0,0) {\textsf{L}};
    \node[anchor=north west, draw=black, thick, minimum height=4cm, minimum width=0.1\linewidth, rounded corners]
        (T) at (TAR) {\textsf{T}};
    \node[anchor=north west, draw=black, thick, minimum height=4cm, minimum width=0.5\linewidth, rounded corners]
        (TAouter) at (0.2\linewidth,0) {};
    \node[anchor=north west, draw=none, thick, minimum height=2.3cm, minimum width=0.5\linewidth, rounded corners]
        (TA) at (0.2\linewidth,-1.6cm) {};
    \node[anchor=north, inner sep=0] (GL) at (0.35\linewidth,-4.1cm) {$\learn_{\textsf{bar}}$};
    \node[fit=(TA)(L)(GL), draw=black, thick, rounded corners, dotted] (BOXGENLEARN) {};

    \draw[-latex, ultra thick] ($(L.north east) + (MQSHIFT)$) -- node[midway, above, scale=0.65] {Is $w \in L_\teach$?} node[pos=0.075, above, scale=0.65] {(\textsf{MQ})} node[pos=0.925, above, scale=0.65] {($\textsf{MQ}_\alpha$)} ($(T.north west) + (MQSHIFT)$);
    \draw[latex-, very thick, dashed] ($(L.north east) + (MASHIFT)$) -- node[midway, below, scale=0.65] {\small\textsc{Yes}/\textsc{No}} ($(T.north west) + (MASHIFT)$);
    
    \draw[-latex, ultra thick] ($(L.north east) + (EQSHIFT)$) -- node[midway, above, scale=0.65] {Is $L_\alpha(\H) = L_\teach$?} node[pos=0.075, above, scale=0.65] {(\textsf{EQ})} node[pos=0.925, above, scale=0.65] {($\textsf{EQ}_\alpha$)}  ($(T.north west) + (EQSHIFT)$);
    \draw[latex-, very thick, dashed] ($(L.north east) + (EASHIFT)$) -- node[midway, above, scale=0.65] {$w_c$} ($(TA.west)$);

    \node[draw, trapezium, anchor=east, trapezium left angle=70, trapezium right angle=110,
        scale=0.65] (COR) at (COMPUTENODE) {Compute $w \equiv_\alpha w_\teach$\\over $\barNames_0$};
    \node[draw, diamond, anchor=east, aspect=1.5, scale=0.65] (D2) at ($(COR.west) + (-0.75,0)$) {$\exists\,w'\alphaequiv w:$\\ $w'\in L(\H)$?};

    \node[anchor=north west,text=gray] at ($(COR.south east) + (14pt,1pt)$) {\bfseries\textsf{1}};
    \node[anchor=south west,text=gray] at ($(D2.north east) + (2pt,-3pt)$) {\bfseries\textsf{2}};
    \node[anchor=center] at ($(COR.center) + (0,0.75cm)$) {\textsf{TA}};

    \draw[-latex, very thick, dashed] ($(T.north west) + (EASHIFT)$) -- node[pos=0.4, above, scale=0.65] {$w_\teach$} ($(COR.east)$);
    \draw[very thick, rounded corners, dashed] ($(T.north west) + (EASHIFT)$) -| ($(COR.east) + (0.1\linewidth,-0.5cm)$);
    \draw[-latex, very thick, rounded corners, dashed] 
        let \p1 = (BOXGENLEARN.east), \p2 = ($(COR.east) + (0,-0.5cm)$) in 
        ($(COR.east) + (0.1\linewidth,-0.25cm)$) |- 
        node[pos=0.7, below, scale=0.65] {\textsc{Yes}} (\x1,\y2);
    
    \draw[-latex, thick, rounded corners]
        (COR.west) -- node[midway, above, scale=0.65] {$w$} (D2.east);

    \draw[-latex, thick, rounded corners] (D2.north) |- node[pos=0.15, right, scale=0.65] {\textsc{Yes}}
        ++(-0.075\linewidth,0.35cm) -| node[pos=-0.2, anchor=east, below, scale=0.65] {$w_c := w'$}
        ($(D2.west) + (-0.025\linewidth,0.5cm)$) |- (TA.west);
    \draw[thick, rounded corners] (D2.south) |- node[pos=0.15, right, scale=0.65] {\textsc{No}}
        ++(-0.075\linewidth,-0.35cm) -| node[pos=-0.2, anchor=east, above, scale=0.65] {$w_c := w$}
        ($(D2.west) + (-0.025\linewidth,-0.5cm)$) |- (TA.west);
\end{tikzpicture}%
}
        \caption{\label{fig:talearner}Learner $\learn_\textsf{bar}$ with internal learner \learn, teaching assistant \tass and teacher \teach. Bold arrows 
        (\thickrarrow\!) 
        denote queries asked by \learn, dashed arrows
        (\dashedrarrow\!)
        denote answers to those queries given by \teach (via \tass) to \learn, and arrows inside of the \tass denote its flow.}
    \end{figure*}
The first component is an \emph{arbitrary} MAT-based learner \learn for standard regular word/$\omega$-/tree languages over $\barNames_0$ that infers a 
bar automaton using a finite number of membership queries $(\textsf{MQ})$ and equivalence queries $(\textsf{EQ})$.  
The learner \learn may, for instance, execute any of the known learning algorithms discussed in \Cref{sec:active-learning}. The specific choice of \learn's algorithm affects the query complexity of~$\learn_{\textsf{bar}}$, but not its correctness and termination.

The second component of $\learn_{\textsf{bar}}$ is a \emph{teaching assistant} (\tass) that internally interacts with~\learn and externally interacts with the teacher \teach. In the interaction with~\learn, the \tass emulates a teacher for the language $L_\teach \restriction \barNames_0$; recall from \Cref{cor:restr-regular} that this is a regular language. 
To achieve this, the \tass turns each query $(\textsf{MQ})$ and $(\textsf{EQ})$ of \learn into a corresponding query $(\textsf{MQ}_\alpha)$ and $(\textsf{EQ}_\alpha)$ for the teacher~\teach. Then the \tass translates \teach's reply back to a form that \learn can process.  
\enlargethispage*{5mm}
In more detail, the \tass proceeds as shown in \Cref{fig:talearner}:
\begin{itemize}
\item A membership query ($\textsf{MQ}$) by $\learn$ (given by a bar input $w\in \barNames_0^*/\barNames_0^\omega/\bartree[\barNamess_0](\Sigma)$) and its answer are relayed
    unchanged between \learn and \teach.
\item For an equivalence query ($\textsf{EQ}$) by \learn with hypothesis $\H$ (a bar automaton over~$\barNames_0$), the query itself is
  relayed unchanged to \teach. If~$\H$ is correct (i.e.~$L_\alpha(\H)=L_\teach$), then the learner~$\learn_{\textsf{bar}}$ successfully terminates. Otherwise \teach
    returns a counterexample, that is, a bar input $w_\teach \in \barAs/\barAw/\bartree(\Sigma)$ in the symmetric difference $L_\alpha(\H)\oplus L_\teach$. 
\item The \tass cannot simply relay the counterexample $w_\teach$ to the internal learner~\learn because the latter expects a counterexample over the finite alphabet $\barNames_0$. Therefore, $w_\teach$ needs to be suitably processed. 
First, the \tass picks any bar input~$w$ over $\barNames_0$ that is $\alpha$-equivalent to~$w_\teach$ (step $\mathsf{1}$ in \Cref{fig:talearner}). 
Such $w$ always exists because the bar languages $L_\alpha(\H)$ and~$L_\teach$ are both accepted by bar automata over $\barNames_0$.
Second, $w$ is chosen to be in the literal language of the hypothesis~$\H$ if possible (step~$\mathsf{2}$ in \Cref{fig:talearner}), and then returned to \learn. 
\end{itemize}

    The effective implementation of steps $\mathsf{1}$ and $\mathsf{2}$ by the \tass is explained in \Cref{prop:complTAS1,prop:complTA} below. First, we establish the correctness
    and the query complexity of the learning algorithm.

    \begin{theorem}[Correctness and Complexity]\label{lem:correctEverything}
      Suppose that the internal learner
      \learn needs ${M(L_\teach\restriction \barNames_0)}$ membership and $E(L_\teach\restriction \barNames_0)$ equivalence queries to learn a
      finite automaton/Büchi automaton/tree automaton for the regular language
      $L_\teach \restriction \barNames_0$. Then $\learn_{\textsf{bar}}$ learns a
      bar automaton for $L_\teach$ with at most $M(L_\teach \restriction \barNames_0)$ membership and $E(L_\teach\restriction \barNames_0)$ equivalence queries.
    \end{theorem}
    \begin{proof}[Sketch]
      Correctness is immediate from termination of $\learn_{\textsf{bar}}$. For termination, one only needs to show that
      the processed counterexamples $w_c$ are elements of the symmetric difference $L(\H) \oplus
      (L_\teach\restriction \barNames_0)$. This follows from the computations in steps 1 and 2 and the closure
      of $L_\teach$ under $\alpha$-equivalence.
    \end{proof}
\begin{expl}
  In the case where $\learn_{\textsf{bar}}$ is instantiated with Angluin's
  $\Lstar$, the complexity of \cref{lem:correctEverything} is as
  follows. Let $k$ be the size of the alphabet $\barNames_0$, and let $n$ be the number of states of the minimal DFA accepting the regular language $L_{\teach} \restriction \barNames_0$. Furthermore, let $m$ be the maximum length of any counterexample given during the run
  of $\learn_{\textsf{bar}}$ by $\teach$. Then, $\learn_{\textsf{bar}}$ (just as~$\Lstar$) asks at
  most $n$ equivalence queries and $\mathcal{O}(kn^2m)$ membership
  queries~\cite[pp.~95--96]{angluin87}. In particular, the total query complexity is polynomial in the three parameters.
\end{expl}

To implement steps $\mathsf{1}$ and $\mathsf{2}$ in \Cref{fig:talearner}, we make use of our procedures for checking $\alpha$-equivalence (\Cref{sec:checkingAlphaEquiv}) and obtain the following complexity:
    \begin{prop} \label{prop:complTAS1}
        For bar word and tree automata, step $\mathsf{1}$ can be implemented in polynomial time. For
        bar Büchi automata, it can be implemented in space exponential in the size of the unknown
        alphabet $\barNames_0$ and the length of $w_\teach$.
    \end{prop}
    \begin{prop} \label{prop:complTA}
      For all bar automata, step $\mathsf{2}$ can be implemented in parametric 
      deterministic linear space, with the size of $\barNames_0$ as the parameter.
    \end{prop}

    \begin{rem}
      \begin{enumerate}
        \item For finite bar strings and bar trees, we can also implement step~$\mathsf{2}$ in non-deterministic polynomial time:
          guess a bar input $w'$ of the same length/shape as~$w$, verify that
          $w'\alphaequiv w$ via their De Bruijn normal forms, and then verify that the bar automaton~$\H$ literally accepts $w'$ by guessing an accepting run.
          These steps could be implemented with a state-of-the-art SAT solver.
        \item As usual for
          active learning methods, the above complexity analysis focuses on the effort of the learner. On the side of teacher, equivalence queries can be implemented by forming closures and then checking for literal language equivalence. This contrasts the case of classical non-deterministic register automata where equivalence is undecidable~\cite{KaminskiFrancez94}.
      \end{enumerate}
    \end{rem}

    \section{Handling Unknown Input Alphabets} \label{sec:alphabet}

    Finally, we discuss how to lift the previous requirement of having a known finite
    alphabet generating the unknown bar language $L_\teach$. The idea is to run the learner~$\learn_{\textsf{bar}}$ with \emph{any} finite alphabet $\barNames_0\seq \barNames$. In case $\learn_{\textsf{bar}}$ gets stuck, it knows that the present alphabet is too small and reboots the learning process with a suitably extended alphabet. Details are as follows.

Initially, $\learn_{\textsf{bar}}$ is executed (running the learning algorithm of \Cref{fig:talearner}) with
the trivial alphabet $\barNames_0=\emptyset$. If a correct hypothesis is found, then $\learn_{\textsf{bar}}$
successfully terminates. If \teach delivers a counterexample $w_\teach$, then it may occur in the computation of step $\mathsf{1}$ that no $w\alphaequiv w_\teach$ over $\barNames_0$ is found. In this case, $\barNames_0$ is extended to a larger finite alphabet  $\barNames_0\seq \barNames_0'\seq \barNames$ such that a representative $w$ of $w_\teach$ over $\barNames_0'$ exists; we explain below how to compute a suitable~$\barNames_0'$. Then~$\learn_{\textsf{bar}}$ restarts with the new alphabet $\barNames_0'$. This process of alphabet extension and restarting~$\learn_{\textsf{bar}}$ is repeated each time step $\mathsf{1}$ gets stuck. At some stage, the current alphabet is large enough to generate the unknown bar language $L_\teach$, and $\learn_{\textsf{bar}}$ runs to completion.

It remains to explain how to extend the current $\barNames_0$ to a larger $\barNames_0'$. In the case of finite bar strings/trees, compute the De Bruijn normal form of $w_\teach$ and choose the extension $\barNames_0\seq \barNames_0'$ to be a smallest possible alphabet such that some $w\alphaequiv w_\teach$ over $\barNames_0'$ exists. The number of bar and plain letters that need to be added to $\barNames_0$ can be read off the normal form.

 In the case of infinite bar strings, construct a Büchi automaton accepting the closure of the ultimately periodic bar string $w_\teach$ with respect to the alphabet $\barNames_0\cup \barNames_1$, where $\barNames_1$ is the set of letters appearing in $w_\teach$, and then search for a smallest subalphabet $\barNames_0\seq\barNames_0'\seq (\barNames_0\cup\barNames_1)$%
 such that this closure contains some ultimately periodic string over~$\barNames_0'$. Note that $\barNames_0'$ contains all free names of $w_\teach$.

    This process of iterated alphabet extension approximates a fitting alphabet from below, ending at a smallest possible $\barNames_0$ that generates $L_\teach$, that is, such that $L_\teach$ is the closure of $L_\teach\restriction \barNames_0$ under $\alpha$-equivalence. To estimate the query complexity of the above algorithm, let us introduce some notation. We write $M(L)$ and $E(L)$ for the number of membership and equivalence queries required by the internal learner \learn to learn a regular language $L$. Moreover, for $k\in \Nat$ we put $M_k = \max \set{M(L_\teach \restriction \barNames_0) \mid \barNames_0\seq \barNames,\, \card{\barNames_0}=k}$; similarly for $E_k$.
\begin{theorem}[Correctness]
Let $n$ be the least cardinality of any finite alphabet $\barNames_0\seq \barNames$ generating $L_\teach$. Then the above extended learning algorithm infers a bar automaton for $L_\teach$ with at most $\sum_{k = 0}^n M_k$ membership queries and $\sum_{k = 0}^n E_k$ 
equivalence queries.
    \end{theorem}

\begin{proof}
Let $n_i$ be the size of the finite alphabet used in the $i$th iteration of $\learn_{\textsf{bar}}$. Since the alphabet grows in every iteration, we have $0= n_0<\ldots<n_k=n$, where $k\!+\!1$ is the number of iterations. The number of membership queries in the $i$th iteration is~at~most~$M_{n_i}$, resulting in at most
$\sum_{i=0}^k M_{n_i} \leq \sum_{i=0}^n M_i$ queries. The same reasoning applies to equivalence~queries.
\end{proof}

    \begin{rem}
        An alternative to restarting the learner $\learn_{\textsf{bar}}$ each time the alphabet grows is to use an internal learner $\learn$ that can handle \emph{dynamically growing input alphabets}, that is,
        can process counterexamples that are not over the current finite alphabet~\cite{ihs13}. The \tass procedure then presents counterexamples to \learn using the minimal number of
        additional letters. The learner \learn
        needs to handle the \emph{grown} alphabet internally and proceed. This method potentially saves queries compared to a  restart but restricts the class of
        possible learners significantly. To the best of our knowledge, learners for dynamically growing alphabets have only been studied in the case of deterministic automata over finite words.
    \end{rem}

    \section{Conclusion and Future Work}\label{sec:concl}

    We have investigated the learnability of bar word languages, bar $\omega$-languages and bar tree languages in Angluin's minimally adequate teacher (MAT) framework. For this purpose, we
    have introduced a learning algorithm $\learn_{\textsf{bar}}$ which builds on an arbitrary MAT-based learner for
    standard word/$\omega$-/tree languages over finite alphabets. Its query complexity is completely determined by that of
    the underlying learner.
    One of the key ingredients  of $\learn_{\mathsf{bar}}$ is an efficient procedure for checking
    $\alpha$-equivalence of bar strings/trees. We have developed a
    version of \emph{De Bruijn levels}~\cite{db72} reducing this problem to checking syntactic
    equality.

When learning bar $\omega$-languages, we 
    currently compute closures of bar automata, which is not needed in the cases of finite bar strings and trees where we instead guess a De Bruijn normal form. It remains to investigate whether a guessing approach is possible in the infinite case; as a prerequisite, this would require a suitable normal form for ultimately periodic bar strings.

    Additionally, a coinductive approach to $\alpha$-equivalence as proposed by Blan\-chette et al.~\cite{bgpt19} would be interesting to study in order to streamline the definition of $\alpha$-equivalence and look for a possible use in defining suitable normal forms.

    Looking at data languages (bar languages under local
    freshness), a further question is whether we can efficently learn such languages without needing an explicit representation
    as bar automata, and with queries asking for membership of data inputs (rather than bar inputs) and equivalence of data languages (rather than bar languages). 
	The key challenge is that the local freshness operator $L\mapsto \LF(L)$ is not injective, so that there is no straightforward back-and-forth translation between bar and data languages.
    
Another open problem is the possibility of conformance testing for bar automata, which is a common approach for implementing equivalence queries in a practical black-box learning setting (e.g.~\cite{DBLP:journals/cacm/Vaandrager17}). 
	Likewise, passive learning of bar languages is an open problem.

	 An alternative to the approach of the present paper would be to develop learning algorithms directly
	at the level of bar automata, or equivalent models in the category of nominal sets (such as RNNAs), instead of reducing
	them to standard learning algorithms over finite alphabets. We expect that categorical
	approaches to automata learning~\cite{us20,hkrss22} would be a good starting point.
      
\clearpage
    \bibliographystyle{splncs04}
    \bibliography{refs}

\clearpage
\appendix
\subsection*{Appendix}
This appendix contains proof details omitted in the paper. In particular, we explain the connection between bar automata and nominal automata.

    \section{Details for~\Cref{sec:prelim}}\label{app:EquivModels}
    
\mypar{More on Nominal Sets}
We recall additional terminology from the theory of nominal sets, which is needed for some of the proofs and constructions presented in the sequel.

    An element $x$ of a nominal set $X$ is \emph{equivariant} if it has an empty support, i.e.~$\supp(x) = \emptyset$, while a subset $X$ of a nominal set $Y$ is \emph{equivariant} if $\pi\cdot x\in X$ for all $x\in X$ and $\pi\in\Perm(\names)$.
    A map $f\colon X \to Y$ between nominal sets is \emph{equivariant} if $f(\pi \cdot x) = \pi \cdot f(x)$ for all $x \in X$ and $\pi \in \Perm(\names)$, which implies $\supp f(x) \subseteq \supp(x)$ for all $x \in X$.
    Similarly, a subset $X$ has \emph{support} $S$ if $\pi \cdot x \in X$ for all $x \in X$ and permutations $\pi$ such that $\pi(a) = a$ for all $a \in S$.
    $X$ is \emph{uniformly finitely supported} if $\bigcup_{x \in X} \supp(x)$ is finite, in which case $X$ is also finitely supported. A uniformly finitely supported equivariant set consists of only elements that are themselves equivariant.

    The coproduct of nominal sets $X_i$ ($i\in I$) is given by their disjoint union $\coprod_{i\in I} X_i = \setw{(i,x)}{i\in I,\, x\in X_i}$ with groups action inherited from the $X_i$, that is, $\pi\cdot (i,x)=(i,\pi\cdot x)$.

    Given a nominal set $X$, the \emph{orbit} of an element $x\in X$ is the set  $\setw{\pi\cdot x}{\pi \in \Perm(\names)}$. The orbits form a partition of $X$.
    A nominal set is \emph{orbit-finite} if it has finitely many orbits. For every finite set $S\seq\names$, an orbit-finite nominal set contains only finitely many elements supported by $S$. The \emph{degree} of an orbit-finite nominal set $X$ is  $\deg(X) = \max_{x\in X} |\supp(x)|$. 

    \begin{expl}
     The nominal set $\Ats$ has infinitely
        many orbits; its equivariant subsets $\names^n$ (words of a fixed length $n$) are orbit-finite. For instance,
        $\names\!^2$ has the two orbits $\{aa: a\in \names\}$ and $\{ab: a\neq b\in \names\}$.
        Another example of an orbit-finite nominal set is \[\names\!^{\#n} = \setw{a_1\ldots a_n}{a_i\neq a_j
        \text{ for $i\neq j$}},\]
an equivariant subset of $\names^n$ with just a single orbit.
        Both $\names^{\# n}$ and $\names^n$ have degree $n$.
    \end{expl}

        \mypar{Details for~\Cref{ex:nonresidual}}  
      We begin by recalling necessary notions to make our argument from~\cite{ms22}:

      \begin{defn}[{\cite{dlt02,ms22}}]
        Given a data language $L \seq \Ats$ and a data word $u \in \Ats$, we define the \emph{derivative of
        $L$ w.r.t.~$u$} by $\deriv{u}{L} := \setw{v}{uv \in L}$. The set of all derivatives is defined by
        $\Der(L) := \setw{\deriv{u}{L}}{u \in \Ats}$.
      \end{defn}
      
      \begin{defn}[{\cite[Def.~2.7]{ms22}}]
        A nominal automaton is \emph{residual} if all states accept derivatives of the language of the
        automaton. A data language $L$ is \emph{residual} if there is a residual nominal automaton accepting $L$.
      \end{defn}
      
      \begin{defn}[{\cite[Def.~4.3]{ms22}}]
        \begin{enumerate}
          \item For some nominal set $X$, we consider the nominal set $\powfs(X)$ of finitely supported subsets
            as a Boolean algebra where all necessary operations ($\wedge$, $\vee$, $\neg$) and the finitely supported join $\bigvee\colon \powfs(\powfs(X)) \to \powfs(X)$ are equivariant.
          \item Let $Y \seq \powfs(X)$ be equivariant and $y \in Y$. Then, $y$ is \emph{join-irreducible in $Y$}
            if $y = \bigvee \mathscr{Y} \implies y \in \mathscr{Y}$ for every finitely supported $\mathscr{Y} \seq Y$.
            The subset of all join-irreducible elements of $Y$ is defined by $\Jir(Y) := \setw{y \in Y}{\text{$y$ is join-irreducible in $Y$}}$.
          \item A subset $Y \seq \powfs(X)$ \emph{generates} a subset $Z \seq \powfs(X)$ if $Z \seq \setw{\bigvee \mathcal{y}}{\text{$\mathcal{y} \seq Y$ f.s.}}$.
        \end{enumerate}
      \end{defn}

      \begin{theorem}[{\cite[Thm.~4.10]{ms22}}]
        A data language $L$ is residual iff the set $\Jir(\Der(L))$ is orbit-finite and generates $\Der(L)$.
      \end{theorem}
      We consider the data language
      \[ \LF(L) = \setw{{u}ab{v}ac{w}bc}{{u}, {v}, {w} \in \Ats,\ a \neq b, b \neq c \in \names}. \]
      There is the following infinite ascending chain of derivatives in the poset $\Der(\LF(L))$:
      \begin{equation} \label{eq:derivsof}
        \LF(L) \subsetneq \deriv{a}{\LF(L)} \subsetneq \deriv{ba}{\LF(L)} \subsetneq \deriv{cba}{\LF(L)} \subsetneq \cdots
      \end{equation}
      To reject residuality of $\LF(L)$, we prove that all derivatives in Equation~\eqref{eq:derivsof} are join-irreducible, whence $\Jir(\Der(L))$ is not orbit-finite.
      \begin{lemma}[Join-Irreducibleness]
        All the derivatives in Equation~\eqref{eq:derivsof} are join-irreducible.
      \end{lemma}
      \begin{proof}
        Consider $w = a_k \cdots a_1 a_0 \in \Ats$ with $k \geqslant 2$ and all $a_i$ distinct.
        We show that $\deriv{w}{\LF(L)}$ is join-irreducible
        in $\Der(\LF(L))$. For this, we notice that $\deriv{u}{%
        \LF(L)} \seq \deriv{w}{\LF(L)}$
        holds iff $u$ is a suffix of $w$. The direction ``$\Leftarrow$'' is simple, since any
        prefix may be skipped ($\Ats$). So suppose then that $u$ is not a prefix of $w$, i.e.~that
        there no $v$ such that $w \neq vu$. Knowing this, there exists an $i \geqslant 0$ with
        $x \neq a_i$ and $u$ contains the suffix $xa_{i - 1}\cdots a_0$. Choose a fresh $a_{-1}$.
        If $x = a_k$ for some $k$, let $c := a_{k - 1}$, otherwise choose $c$ fresh. 
        
        With this
        choice, we see that $a_{-1}xca_{i-1}c \in \deriv{u}{\LF(L)}$.
        We then show $a_{-1}xca_{i-1}c \notin \deriv{w}{\LF(L)}$
        by exhaustion: If $x$ does not occur in $w$, then $c$ is fresh for $w$, which means that
        all letters in $wa_{-1}xca_{i-1}c$ are unique except for $a_{i - 1}$ which occurs twice
        even with different successors, which do not occur after the last pair with $a_{i - 1}$.
        Therefore, $a_{-1}xca_{i-1}c \notin \deriv{w}{\LF(L)}$.
        If $x = a_k$ for some $k$, and therefore $c = a_{k - 1}$, then $wa_{-1}xca_{i-1}c$ mentions
        $a_k$ and $a_{i - 1}$ twice, and $a_{k - 1}$ trice, however not in the form required by
        $\LF(L)$, since the different successors of the repeated letters
        are not repeated. Therefore, $a_{-1}xca_{i-1}c \notin \deriv{w}{\LF(L)}$. We have thus shown the following:
        \[%
            \setw{u}{\deriv{u}{\LF(L)} \seq
            \deriv{w}{\LF(L)}} =
            \setw{u}{u \text{ is a suffix of } w}.
        \]
        Now consider $X = \bigvee\setw{\deriv{u}{\LF(L)}}{u\text{ is a
        strict suffix of } w}$ and see that $X \neq \deriv{w}{\LF(L)}$,
        since $a_ka_0a_{k-1}a_0 \in \deriv{w}{\LF(L)} \setminus X$,
        thereby making $\deriv{w}{\LF(L)}$ join-irreducible.
      \end{proof}
      \begin{corollary}
        The language $\LF(L)$ is not residual.
      \end{corollary}
      \begin{rem}
        There is another notion of \emph{non-guessing residuality} mentioned by~\cite{ms22},
        which is a subset of the aformentioned residuality requiring the residual automaton to be
        \emph{non-guessing}. While $\LF(L)$ is accepted by a \emph{non-guessing} nominal automaton (\Cref{sec:bar-automata}),
        it cannot be accepted by a (non-guessing) residual automaton. 
        Overall, the modified version of \nomNLstar \cite{ms22} cannot ensure termination for
        the data language $\LF(L)$ due to its non-residuality.
      \end{rem}

    \mypar{Proof of~\Cref{prop:barautclosure}}

    We construct for every type of bar automaton its closure directly. All constructions here are equivalent to a
    back-and-forth translation between bar automata and their equi-expressive nominal variants. 
    
    \begin{defn}\label{defn:dollarSets}
      For $n\in\Nat$, we write $\mathbf{n}=\{1,\dots,n\}$. We denote by
      $X^{\$\mathbf{n}}$ the set of all \emph{partial injective maps}
      from $\mathbf{n}$ to the set $X$.
      The \emph{domain} $\dom(r)$ of $r$
      is the set of all $x \in \mathbf{n}$ for which $r(x)$ is defined.
      A partial injective map
      $\overline{r} \in X^{\$\mathbf{n}}$ \emph{extends} $r \in X^{\$\mathbf{n}}$, denoted
      $r \leqslant \overline{r}$, if
      $\dom(r) \seq \dom(\overline{r})$ and~$r$ and~$\overline{r}$
      coincide on $\dom(r)$, that is, $r(x) = \overline{r}(x)$ for all
      $x \in \dom(r)$.
    \end{defn}
    
    \begin{fact}
      For finite sets $X$, we have $\card{X^{\$\mathbf{k}}} \leqslant \card{\mathbf{k} \rightharpoonup X} = (1 + \card{X})^k = 2^{k \log(1 + |X|)}$.
    \end{fact}

    \paragraph*{Bar Word Automata} The idea behind the construction is to regard the original bar word automaton
    as accepting \enquote{patterns} and to match those to bar strings over the finite alphabet by the use of a
    register automaton (wherein the registers are restricted to the finite alphabet).
    
    \begin{construction}\label{constr:closeNFA}
      \sloppypar
      Given a bar word automaton $A = \maketuple{Q, \barNames_0,\to, q_0, F}$ such that 
      ${k := \card{\barNames_0 \cap \names}}$ and $\names_0 := \barNames_0 \cap \names = \set{a_1, \dots, a_k}$,
      we construct the bar NFA $\barA = \maketuple{\overline{Q}, \barNames_0, \to', \overline{q_0}, \overline{F}}$
      as follows:
      \begin{itemize}
        \item $\overline{Q} := Q \times \parnom[0]{k}$;
        \item $\overline{q_0} = \maketuple{q_0, (i \mapsto a_i)}$;
        \item $\overline{F} := F \times \parnom[0]{k}$; and
        \item transitions $\to'$:
        \begin{itemize}
          \item $\maketuple{q, r} \xra{a}\!'{}\, \maketuple{q', r'}$ iff there is an $i \leqslant k$ such that
            $r(i) = a$; $r' \leqslant r$; and $q \xra{a_i} q'$.
          \item $\maketuple{q, r} \xra{\scriptnew{a}}\!'{}\, \maketuple{q', r'}$ iff there is an
            $\alpha \in \barNames_0\setminus\names$ such that for all $i \leqslant k$, (i) if $r'(i) = a$, then $\alpha = \newletter{a_i}$ and (ii) if $\alpha = \newletter{a_i}$ and $i \in \dom(r')$, then $r'(i) = a$; (iii) for all $i \in \dom(r')$, $r'(i) \in \set{a,r(i)}$; and (iv) $q \xra{\alpha} q'$.
        \end{itemize}
      \end{itemize}
    \end{construction}

    \begin{lemma}\label{lem:corrCloseNFA}
      Let $A$ be a bar word automaton and $\barA$ the bar NFA of~\Cref{constr:closeNFA}. Then $L(\barA) =
      L_\alpha(A) \restriction \barNames_0$ and $L_\alpha(\barA) = L_\alpha(A)$.
    \end{lemma}
    \begin{proof}
      {}\sloppypar
      Given $A = \maketuple{Q, \barNames_0, \to, q_0, F}$, $\names_0 := \barNames_0\cap\names = \set{a_1, \dots, a_k}$
      and $\barA = \maketuple{\overline{Q}, \barNames_0, \to', \overline{q_0}, \overline{F}}$, we see that equality
      of bar languages follows directly from the first property by fact that both $A$ and $\barA$ are over the same
      finite alphabet. Regarding $L(\barA) = L_\alpha(A) \restriction \barNames_0$, we show mutal inclusion:

      \noindent\enquote{$\qes$}: Let $w = \alpha_1\cdots\alpha_n \alphaequiv w' = \alpha_1'\cdots\alpha_n' \in L(A)$,
      $w \in \barAs_0$ and $q_0 \xra{\alpha_1'} q_1 \xra{\alpha_2'} \cdots \xra{\alpha_n'} q_n$ be an accepting run for
      $w'$ in $A$. We define the partial injective maps $r_0\colon \mathbf{k} \rightharpoonup \names_0, i \mapsto a_i$
      and
      \[
        r_i\colon \mathbf{k} \rightharpoonup \names_0, j \mapsto \begin{cases} \ub(\alpha_i) & \text{, if $\alpha_i' = \newletter{a_j}$ and $\ub(\alpha) \in \FN(\alpha_{i + 1}\cdots\alpha_n)$} \\
          r_{i - 1}(j) & \text{, if $\alpha_i \neq \newletter{r_{i-1}(j)}$ and $r_{i-1}(j) \in \FN(\alpha_{i + 1}\cdots\alpha_n)$} \\
          \bot & \text{, otherwise.} \end{cases} \quad (1 \leqslant i \leqslant n)
      \]
      This results in a run $\maketuple{q_0,r_0} \xra{\alpha_1}\!'{} \maketuple{q_1,r_1} \xra{\alpha_2}\!'{} \cdots
      \xra{\alpha_n}\!'{} \maketuple{q_n,r_n}$ for $w$ in $\barA$, where acceptance is clear by construction of
      $\barA$. Lastly, we show that all proposed transitions do exist: Take any $\maketuple{q_i,r_i}
      \xra{\alpha_{i+1}}\!'{} \maketuple{q_{i+1},r_{i+1}}$. Existence if $\alpha_{i+1} \notin \names$ is obvious, so
      suppose further that $\alpha_{i+1} \in \names_0$. Because $w \alphaequiv w'$, we see that their de Bruijn normal
      form $\nf(w) = \beta_1\cdots\beta_n$ coincides. With $\alpha_{i+1}' = a_j \in \names_0$, we look at the (unique)
      index $k \leqslant i$ such that $\alpha_k' = \newletter{a_j}$ and $\beta_k = \beta_{i+1}$ We notice immediately
      that $r_{k}(j) = \alpha_{i+1}$, since $\alpha_{i+1} \in \FN(\alpha_{k+1}\cdots\alpha_n)$ (by $\alpha$-equivalence
      and properties of the de Bruijn normal form). Additionally, for all $k < m \leqslant i$, we have $r_m(j) = 
      r_{m-1}(j) = \alpha_{i+1}$ because $\alpha_i \in \FN(\alpha_m\cdots\alpha_n)$ and $\alpha_{m-1} \neq \newletter{\alpha_i}$ (this would contradict the required $\alpha$-equivalence). Thus, $w \in L(\barA)$.

      \noindent\enquote{$\seq$}: 
      Let $\maketuple{q_0, r_0} \xra{\alpha_1}\!'{}\, \maketuple{q_1, r_1} \xra{\alpha_2}\!'{}\, \cdots
      \xra{\alpha_n}\!'{}\, \maketuple{q_n,r_n}$ be an accepted run of $w = \alpha_1\cdots\alpha_n \in \barAs_0$
      in $\barA$. By construction, there is a corresponding transition $q_{i-1} \xra{\alpha_i'} q_i$ in $A$
      for every $1 \leqslant i \leqslant n$. The precise choice of these $\alpha_i'$s is irrelevant, that is, we
      only assume that every $\alpha_i'$ witnesses the existential quantifier in the condition of the original 
      transition. This resulting run is also accepting.      
      
      Note that this choice is miniscule and only affects those transitions and bar names not immediately stored
      in the subsequent \enquote{register}: Indeed, if $\alpha_i \in \names$, there must be some $j \leqslant k$ with
      $r_i(j) = \alpha_i$, thereby fixing $\alpha_i' = a_j$. If $\alpha_i \notin\names$ and $r_{i}(j) =
      \ub(\alpha_i)$ holds for some $j \leqslant k$, then
      $\alpha_i' = \newletter{a_j}$ is fixed (by condition (ii)). Lastly, if $\alpha_i \notin\names$ and $\ub(\alpha_i)$
      is unequal to every $r_{i}(j)$, we see that $\ub(\alpha_i)$ also does not occur freely in
      $\alpha_{i + 1}\cdots\alpha_n$ by a direct consequence of the definition of plain transitions, thus giving
      an inconsequential choice overall.

      Lastly, we show $\alpha$-equivalence of $w$ and $w' = \alpha_1'\cdots\alpha_n'$ by equality of their de Bruijn
      normal forms (\Cref{lem:aeDBNF}). We show, that for every $n \in \Nat$, if $w = \alpha_1\cdots\alpha_n$ has
      a run in $\barA$ and $w' = \alpha_1'\cdots\alpha_n'$ is the corresponding word with run in $A$, then $w
      \alphaequiv w'$. This via induction over $n$: The base case ($n = 0$) is clear.
      For the induction step, assume that $v = \alpha_1\cdots\alpha_{n} \alphaequiv \alpha_1'\cdots\alpha_{n}' := v'$ by
      use of the induction hypothesis. We show $v\alpha_{n+1} \alphaequiv v'\alpha_{n+1}$. We consider the cases where
      $\alpha_{n+1} \notin\names$ and $\alpha_{n+1} \in \names$: The former case holds trivially by definition
      and~\Cref{lem:aeDBNF}.
      
      For the latter case, if $\alpha_{n+1} \in \names$, i.e.~
      $\maketuple{q_n, r_n} \xra{\alpha_{n+1}}\!'{}\, \maketuple{q_{n + 1}, r_{n + 1}}$ and
      $\alpha_{n+1} = b \in \names_0$, there must be some $j \leqslant k$, such that $r_n(j) = b$ and
      $q_n \xra{a_j} q_{n + 1}$, as well as $r_{n + 1} \leqslant r_n$. Thus, $\alpha_{n+1}' = a_j$. There are now two possibilities:
      \begin{enumerate}
        \item There has been a prior $\newletter{b}$-transition with a corresponding $\newletter{a_j}$-transition.
          Then, both indices in the de Bruijn normal form (\Cref{lem:aeDBNF}) of $v\alpha_{n+1}$ and $v'\alpha_{n+1}'$
          must be equal (namely to the index of the last such corresponding transition).
        \item There has been no prior $\newletter{b}$-transition with corresponding $\newletter{a_j}$-transitions.
          Thus, $b = a_j = r_0(j)$ and $a_j$ is free in $v\alpha_{n+1}$. $a_j$ is then also free in $v'\alpha_{n+1}'$,
          since by condition (ii) of bar transitions in $\barA$, if there had been an
          $\alpha'_{k} = \newletter{a_j}$ with $k \leqslant n$, the corresponding $\alpha_k = \newletter{c}$ would
          have been stored in $r_k(j)$ and all subsequent $r_{i}(j)$'s. This results in the desired $\alpha$-equivalence.
      \end{enumerate}
      By induction, we see that $w \alphaequiv w'$, therefore $w \in L_\alpha(A)$ and finally $L(\barA) =
      L_\alpha(A) \restriction \barNames_0$.
    \end{proof}

    \paragraph*{Bar Büchi Automata} Let $A$ be the bar Büchi automaton over $\barNames_0$ and $\barA$ the
    bar NFA of~\Cref{constr:closeNFA} interpreted with the Büchi semantics. It is readily verified
    that the arguments of~\Cref{lem:corrCloseNFA} still hold under the changed Büchi semantics. Therefore,
    $\barA$ is the closed bar Büchi automaton to $A$.

    \paragraph*{Bar Tree Automata} While the previously discussed bar tree automata were bottom-up tree
    automata, this section will deal with top-down tree automata. It is well-known that top-down and bottom-up NFTAs
    are expressively equivalent (see e.g.~\cite[Thm.~1.6.1]{tata08}) by reversing all rewrite rules, swapping final and 
    initial states, and reducing those to a single one and that bottom-up NFTAs are determinisable 
    (see e.g.~\cite[Thm.~1.1.9]{tata08}), i.e.~bottom-up NFTAs and bottom-up DFTAs are expressively equivalent by use 
    of the powerset construction. So while there is no difference in expressivity, a direct translation to a closed 
    bottom-up tree automaton counteracts the similarity between and intuition behind the various closure constructions
    for bar automata. 
    
    A bar (top-down) non-deterministic finite tree automaton $A = \maketuple{Q, \barNames_0, \Sigma, q_0, \Delta}$ has
    a finite set $Q$ of \emph{states}, an \emph{initial state} $q_0$ and a transition relation $\Delta \seq
    Q \times \barNames_0 \times (\coprod_{\nicefrac{f}{n} \in \Sigma} Q^n)$ consisting of \emph{rewrite rules}
    of the form \[ q(\alpha.f(x_1, \dots, x_n)) \to \alpha.f(q_1(x_1), \dots, q_n(x_n)). \qquad 
      (\alpha \in \barNames_0, \nicefrac{f}{n} \in \Sigma) \]
    A \emph{run} for a tree $t \in \bartree[\barNamess_0](\Sigma)$ from $q \in Q$ is a tree $\runtree(t)$ over the
    signature $Q \times \barNames_0 \times \Sigma$ subject to the following conditions:
    \begin{enumerate}
      \item Every node $\alpha.f(t_1, \dots, t_n)$ in $t$ corresponds to one node
        $q.\alpha.f(t_1', \dots,  t_n')$ in $\runtree(t)$ at the same position (i.e.~in the same context) and
        vice versa. Structurally, both $t$ and $\runtree(t)$ are identical.\label{run:cond:A}
      \item Every node $q.\alpha.f(q_1.t_1, \dots, q_n.t_n)$ in $\runtree(t)$ has a corresponding rewrite
        rule\label{run:cond:B} \[q(\alpha.f(x_1, \dots, x_n)) \to \alpha.f(q_1(x_1), \dots, q_n(x_n)) \in \Delta.\]
      \item The root node is annotated by $q$.\label{run:cond:C}
    \end{enumerate}
    A tree $t \in \bartree[\barNamess_0](\Sigma)$ is accepted by $A$ iff there is a run for $t$ from $q_0$.
    \begin{construction}\label{constr:closeNFTA}
      \sloppypar
      Given a top-down bar tree automaton $A = \maketuple{Q, \barNames_0, \Sigma, q_0, \Delta}$ such that 
      ${k := \card{\barNames_0 \cap \names}}$ and $\names_0 := \barNames_0 \cap \names = \set{a_1, \dots, a_k}$,
      we construct the bar NFTA $\barA = \maketuple{\overline{Q}, \barNames_0, \Sigma, \overline{q_0},
      \overline{\Delta}}$ as follows:
      \begin{itemize}
        \item $\overline{Q} := Q \times \parnom[0]{k}$;
        \item $\overline{q_0} = \maketuple{q_0, (i \mapsto a_i)}$; and
        \item rewrite rules $\overline{\Delta}$:
        \begin{itemize}
          \item $\maketuple{q,r}(a.f(x_1, \dots, x_n)) \to a.f(\maketuple{q_1,r_1}(x_1),\dots,\maketuple{q_n,r_n}(x_n))
            \in \overline{\Delta}$ iff there is an $i \leqslant k$ such that $r(i) = a$; $r_j \leqslant r$ for all
            $j \leqslant n$; and $q(a.f(x_1, \dots, x_n)) \to a.f(q_1(x_1),\dots,q_n(x_n)) \in \Delta$.
          \item $\maketuple{q,r}(\newletter{a}.f(x_1, \dots, x_n)) \to \newletter{a}.f(\maketuple{q_1,r_1}(x_1),\dots,
            \maketuple{q_n,r_n}(x_n)) \in \overline{\Delta}$ iff there is an $\alpha \in \barNames_0\setminus\names$
            such that for all $j \leqslant n$ and all $i \leqslant k$, $\ell \in \dom(r_j)$, we have (i) $r_j(i) = a
            \implies \alpha = \newletter{a_i}$; (ii) $(\alpha = \newletter{a_i} \wedge i \in \dom(r_j)) \implies
            r_j(i) = a$; (iii) $r_j(\ell) \in \set{a, r(\ell)}$; and (iv) $q(\newletter{a}.f(x_1, \dots, x_n)) \to
            \newletter{a}.f(q_1(x_1),\dots,q_n(x_n)) \in \Delta$.
        \end{itemize}
      \end{itemize}
    \end{construction}

    \begin{lemma}\label{lem:corrCloseNFTA}
      Let $A$ be a top-down bar tree automaton and $\barA$ the bar NFTA of~\Cref{constr:closeNFTA}. Then
      $L(\barA) = L_\alpha(A) \restriction \barNames_0$ and $L_\alpha(\barA) = L_\alpha(A)$.
    \end{lemma}
    \begin{proof}
      Given $A = \maketuple{Q, \barNames_0, \Sigma, q_0, \Delta}$,
      $\names_0 := \barNames_0\cap\names = \set{a_1, \dots, a_k}$
      and $\barA = \maketuple{\overline{Q}, \barNames_0, \Sigma, \overline{q_0}, \overline{\Delta}}$, we see that
      equality of bar languages follows directly from the first property by fact that both $A$ and $\barA$ are over the
      same finite alphabet. Regarding $L(\barA) = L_\alpha(A) \restriction \barNames_0$, we show mutal inclusion:

      \noindent\enquote{$\qes$}: Let $t \in L_\alpha(A) \restriction \barNames_0$, that is,
      $t \alphaequiv t' \in L(A)$, $t \in \bartree[\barNamess_0](\Sigma)$ and $\runtree(t')$
      be an accepting run of $t'$. We translate this run into a run for $t$ in $\barA$ recursively as follows:
      \[ \begingroup\setlength{\arraycolsep}{2pt}\begin{array}{r@{\,}rcl}
        \mathsf{tr}\colon & \parnom[0]{k} \times \bartree[\barNamess_0](\Sigma) \times 
          \bartree[Q \times \barNamess_0](\Sigma) & \longrightarrow & 
          \bartree[\overline{Q} \times \barNamess_0](\Sigma)\\
                & \maketuple{r, \alpha.f(t_1, \dots, t_n), q(\alpha'.f(t'_1, \dots, t'_n))} & \longmapsto &
                \displaystyle \maketuple{q,r}(\alpha.f(\mathsf{tr}(r_1, t_1, t'_1), \dots, \mathsf{tr}(r_n,t_n,t'_n)))
        \end{array}\endgroup \]
      Herein, the $r_i \in \parnom[0]{k}$ ($1 \leqslant i \leqslant n$) are defined pointwise, we have
      \begin{equation} r_i\colon \mathbf{k} \to \names, j \mapsto \begin{cases}
        \ub(\alpha) & \text{, if $\alpha' = \newletter{a_j}$ and $a_j \in \FN(t_i')$} \\
        r(j) & \text{, if $\alpha \neq \newletter{r(j)}$ and $r(j) \in \FN(t_i)$} \\
        \bot & \text{, otherwise.}
      \end{cases} \label{defn:ris}\end{equation}
      Given $\overline{q_0} = \maketuple{q_0, r_0}$, we readily verify that $\mathsf{tr}(r_0, t, \runtree(t'))$
      is a valid accepting run for $t$ from $\overline{q_0}$. Conditions (i) and (iii) are satisfied by construction.
      Regarding condition (ii), i.e.~that every node has a corresponding rewrite rule, we take any node
      $\maketuple{q,r}.(\alpha.f(\maketuple{q_1,r_1}\widetilde{t}_1, \dots, \maketuple{q_n,r_n}\widetilde{t}_n))$ in
      $\mathsf{tr}(r_0, t, \runtree(t'))$ and look at the following two cases: Whenever $\alpha \in \barNames_0\setminus
      \names$, it is obvious that $\maketuple{q,r}(\alpha.f(x_1, \dots, x_n)) \to \alpha.f(\maketuple{q_1,r_1}(x_1),
      \dots, \maketuple{q_n,r_n}(x_n)) \in \overline{\Delta}$. So suppose $\alpha \in \names_0$ and $\alpha' = a_i$.
      We immediately see that $r(i) = \alpha$ by a direct consequence of~\eqref{defn:ris} and $\alpha$-equivalence.
      The proof of this is analogous to the (in)finite word case. But then, we also see that $\maketuple{q,r}(
        \alpha.f(x_1, \dots, x_n)) \to \alpha.f(\maketuple{q_1,r_1}(x_1), \dots, \maketuple{q_n,r_n}(x_n)) \in
        \overline{\Delta}$. Therefore, $t \in L(\barA)$.

      \noindent\enquote{$\seq$}: Let $\runtree(t)$ be an accepted run of $t \in \bartree[\barNamess_0](\Sigma)$ in
      $\barA$. By the definition of runs, we know that every node
      $\maketuple{q,r}(\alpha.f(\maketuple{q_1,r_1}t_1, \dots, \maketuple{q_n,r_n}t_n))$ in $\runtree(t)$ corresponds
      to a rewrite rule
      \[ \maketuple{q,r}(\alpha.f(x_1, \dots, x_n)) \to \alpha.f(\maketuple{q_1,r_1}(x_1), \dots, \maketuple{q_n,r_n}
        (x_n)) \in \overline{\Delta}. \]
      By construction of $\barA$, these rewrite rules correspond to rewrite rules
      \[ q(\alpha'.f(x_1, \dots, x_n)) \to \alpha'.f(q_1(x_1), \dots, q_n(x_n)) \in \Delta \text{ for some
        $\alpha' \in \barNames_0$.} \]
      These $\alpha'$ witness the respective existential quantifier. Using this, we get a tree $\widetilde{t'} \in
      \bartree[Q \times \barNamess_0](\Sigma)$ which is a run in $A$ for its underlying tree $t' \in \bartree(\Sigma)$.
      $\alpha$-Equivalence of $t$ and $t'$ is then shown by equality of their de Bruijn normal forms
      (\Cref{lem:aeTrees}). Since the de Bruijn normal forms are essentially equal to branchwise finite word de Bruijn 
      normal forms, it is enough to check equality of those branchwise normal forms. This is readily verified with an
      analogous argument as in the finite word case (\Cref{lem:corrCloseNFA}) because the rules for transitions in
      $\barA$ is branchwise equivalent to the rules for transitions in the finite word case.
    \end{proof}

    \paragraph*{Blow-Up of Constructions} Regarding the overall complexity: Let $\barNames_0$ be of cardinality
    $k_0$, $k_1 := \card{\barNames_0 \cap \names} < k_0$ and $n$ be the number of states of the bar automaton $A$.
    All constructions result in an automaton with at most $n \cdot 2^{k_1\log(k_1 + 1)}$ states. 
    This follows directly from the definition as an $n$-fold coproduct of the set of partial injective functions
    $(\barNames_0 \cap \names)^{\$\mathbf{k_1}}$.

    \begin{rem}
      Lastly, we remark about the generality of the proven statement: Note that for both bar word automata and bar tree
      automata, the resulting closed bar automaton is always non-deterministic, and thus in turn determinisable.
      However, this non-determinism does not have any impact on the algorithms used for learning.
      Indeed, both non-emptiness as well as the \enquote{product automaton}-construction are still in polynomial time
      even if the bar automaton is non-deterministic. (see~\Cref{prop:complTA})
    \end{rem}
    
    \mypar{Details for~\Cref{rem:tightness}} Fix $k, n \in \Nat$. For
    $\barNames_0 = \setw{a_i, \newletter{a_i}}{1 \leqslant i \leqslant k}$, we look at the bar
    automaton $\A$ accepting only the bar string
    $w = \newletter{a_1}\cdots\newletter{a_k}a_1^n\cdots a_k^n$.  Note that~$\A$ has at least
    $(n + 1)\cdot k$ states. The closed bar automaton $\A_{\cl}$ accepts precisely the bar
    strings of the form $\newletter{a_{j_1}}\cdots\newletter{a_{j_k}}a_{j_1}^n\cdots a_{j_k}^n$,
    where $\maketuple{j_1,\dots,j_k}$ is any permutation of $\maketuple{1,\dots,k}$.  Next we show
    that $\A_{\cl}$ has at least one state for every
    $1 \leqslant \ell \leqslant n, 1 \leqslant i \leqslant k$ and every choice of elements
    $a_{j_i}, \dots, a_{j_k} \in \set{a_1, \dots, a_k}$ which accepts the bar
    string $a_{j_i}^\ell\cdots a_{j_k}^n$. Indeed, if that were not the case, then $\A_{\cl}$ would either
    have one state $q$ accepting two different bar strings $a_{j_i}^\ell\cdots a_{j_k}^n$ and
    $a_{j_i'}^\ell\cdots a_{j_k'}^n$ or one state $q$ accepting both
    $a_{j_i}^\ell\cdots a_{j_k}^n$ and $a_{j_i}^{\ell'}\cdots a_{j_k}^n$ for $\ell \neq
    \ell'$. In the former case, the state $q$ is, by definition of the accepted language of
    $\A_{\cl}$, reached by
    $\newletter{a_{j_1}}\cdots \newletter{a_{j_k}}a_{j_1}^n\cdots a_{j_i}^{n-\ell}$, which
    implies that $\A_{\cl}$ accepts the bar string
    $\newletter{a_{j_1}}\cdots\newletter{a_{j_k}}a_{j_1}^n\cdots
    a_{j_i}^{n-\ell}a_{j_i'}^\ell\cdots a_{j_k'}^n$, a clear contradiction. In the latter case,
    we see that the state $q$ is again reached by
    $\newletter{a_{j_1}}\cdots\newletter{a_{j_k}}a_{j_1}^n\cdots a_{j_i}^{n-\ell}$, which
    implies acceptance of the bar string
    $\newletter{a_{j_1}}\cdots\newletter{a_{j_k}}a_{j_1}^n\cdots a_{j_i}^{n-\ell+\ell'} \cdots
    a_{j_k'}^n$ with $n - \ell + \ell' \neq n$, which again is a contradiction.  This results in at
    least $\sum_{i = 1}^{k} n \cdot \frac{k!}{i!}$ states, because there are
    $\nicefrac{k!}{i!}$ choices for the $a_{j_i}, \dots, a_{j_k}$. This number of states is in
    $\Theta(n\cdot k!)$ as shown by the following inequations: 
    \[ 
      n \cdot k!
      \leqslant
      n \cdot k! \cdot \sum_{i = 1}^{k} \frac{1}{i!}
      <      
      n \cdot k! \cdot \sum_{i=1}^\infty\frac{1}{i!}
      \leqslant
      (e-1) \cdot n \cdot k!
      <
      2 \cdot n \cdot k!,
    \]
    using the power series expansion of the exponential function $e^x = \sum_{i=0}^\infty x^i/i!$.

 \section{Details for \Cref{sec:checkingAlphaEquiv}}
We prove that the normal form for bar strings 
    is stable under permutations and concatenations to the left:
    \begin{lemma}\label{lem:nfstbPP}
        Let $w \in \barAs$ be a bar string and $\pi \in \Perm(\names)$ a permutation. Then
        $\nf(\pi \cdot w) = \pi \cdot \nf(w)$, where permutations act trivially on natural numbers.
    \end{lemma}
    \begin{proof}
        Let $w = \alpha_1 \cdots \alpha_n$, $\nf(w) = \beta_1 \cdots \beta_n$ and $\nf(\pi \cdot w) = \gamma_1 \cdots
        \gamma_n$ be defined as above. We show that $\gamma_i = \pi \cdot \beta_i$ holds for all $i$ by case
        distinction: If $\beta_i \in \names$, then $\alpha_i$ is free in $w$ as is $\pi(\alpha_i)$ in $\pi \cdot w$.
        Therefore, $\gamma_i = \pi(\alpha_i) = \pi(\beta_i)$.
        If $\beta_i = k$ and $k$ does not occur in $\beta_1 \cdots \beta_{i - 1}$, then $\alpha_i = \newletter{a} \in
        \barNames$ and $\alpha_1 \cdots \alpha_{i - 1}$ contain $k - 1$ bar names. However, so does
        $\pi \cdot (\alpha_1 \cdots \alpha_{i - 1})$ and $\pi\cdot \alpha_i = \newletter{\pi(a)} \in \barNames$.
        Therefore, $\gamma_i = k = \pi \cdot \beta_i$.
        Lastly, if $\beta_i = k$ and $k$ occurs in $\beta_1 \cdots \beta_{i - 1}$, let $j < i$ denote the smallest index
        with $\beta_j = k$. Existence of such a $j$ is obvious by design. Additionally, $\alpha_i \in \names$ and
        $\alpha_j = \newletter{\alpha_i}$.
        Then, $\pi(\alpha_i) \in \names$ and $\pi \cdot \alpha_j = \newletter{\pi(\alpha_i)}$ and there is no
        index $\ell > j$ with $\pi \cdot \alpha_\ell = \newletter{\pi(\alpha_i)}$. Therefore, $\gamma_j = k$
        (by the argument before) and $\gamma_i = k$. 
      \end{proof}
      \begin{rem}\label{R:pi.w}
        Note that for a permutation $\pi$ that does not change any free letters of the word $w$,
        we have $\nf(w) = \nf(\pi\cdot w)$.
      \end{rem}
      \begin{lemma}\label{lem:nfstbCC}
        Let $v, w, x \in \barAs$ be bar strings. If $\nf(w) = \nf(v)$, then $\nf(xw) = \nf(xv)$.
    \end{lemma}
    \begin{proof}
        Let $w = \alpha_1 \cdots \alpha_n$, $x = \beta_1 \cdots \beta_k$, $\nf(x) = \gamma_1 \cdots \gamma_k$
        and $\nf(w) = \gamma_1' \cdots \gamma_n'$. Then $\nf(xw) = \delta_1 \cdots \delta_{n + k}$ is given as follows:
        For $1 \leqslant i \leqslant k$, we have $\delta_i = \gamma_i$. Let $\ell = \max
        \setw{\gamma_i \in \Nat}{1 \leqslant i \leqslant k}$. If $\gamma_{i - k}' \in \Nat$, then $\delta_i = 
        \gamma_{i - k}' + \ell$. This is easily seen, since if $\gamma_{i - k}' = j$ and if $j$ does not occur in
        $\gamma_1' \cdots \gamma_{i-k-1}'$, there is a bar name in $w$ at position $i - k$ and $j - 1$ bar names
        occur prior to it. Since the prefix $x$ has $\ell$ bar names, there is still a bar name in $xw$ at position
        $i$ but now there are $j - 1 + \ell$ bar names prior to this position. Similarly, if $\gamma_{i - k}' = j$ occurs
        in $\gamma_1' \cdots \gamma_{i-k-1}'$ first at position $m < i - k$, then there is a plain name in $w$ at position
        $i - k$ that corresponds to the bar name at position $m$ such that this bar name does not occur between
        positions $m$ and $i - k$. The prefix again does not change anything in this situation besides increasing the
        level from $j$ to $j + \ell$ because there are $\ell$ more bar names before position $m$. This makes the above
        definition the De Bruin normal form for $xw$, which depends only on $\nf(w)$ and $\nf(x)$. Therefore, $\nf(xw)
        = \nf(xv)$, since $\nf(w) = \nf(v)$.
    \end{proof}
    
\mypar{Proof of \Cref{lem:aeDBNF}}
        \noindent $(\Rightarrow)$: Let $w \alphaequiv v$ be in one step, that is, $w = \newletter{a}w' 
        \alphaequiv \newletter{b}v' = v$ by $\braket{a}w' = \braket{b}v'$. This is without loss of generality, since syntactic
        equality is transitive and prefixes do not change equality of normal forms (\Cref{lem:nfstbCC}).
        Thus, we have some fresh $c \in \names$ such that $\makecycle{a,c} \cdot w' = \makecycle{b,c} \cdot v'$
        and also $\makecycle{a,c} \cdot w = \makecycle{b,c} \cdot v$. By~\Cref{lem:nfstbPP} and
        \Cref{R:pi.w}, we see that
        $\nf(w) = \nf(\makecycle{a,c} \cdot w) = \nf(\makecycle{b,c} \cdot v) = \nf(v)$.

        \noindent $(\Leftarrow)$: We first show that given bar strings $w = xu$ and $v = xy$
        with longest joint prefix $x$ and identical normal form
        $\nf(w) = \nf(v) = \gamma_1 \cdots \gamma_n$, either $w = v$ or two
        $\alpha$-equivalent $w' = x'u'$ and $v' = x'y'$ exist, i.e.~$w \alphaequiv w'$ and
        $v \alphaequiv v'$, such that $\abs{x'} > \abs{x}$.  This implies the statement using
        an iterative process. If $\abs{x} = \abs{w}$, then $w = v$ by design, so assume
        $\abs{x} < \abs{w}$ and $w \neq v$. Because $x$ is the longest joint prefix of $w$ and
        $v$, and both normal forms are equal ($\nf(w) = \nf(v)$), we see that both $u$ and $y$
        start with different bar names, say $\newletter{a}$ and $\newletter{b}$ for $a \neq
        b$. Let the level, i.e.~the value in the normal form, of those bar names be
        $\ell$. Let $c \in \names$ be a fresh name for $w$ and $v$, then both
        $xu \alphaequiv x \left(\makecycle{a,c}\cdot u\right)$ and
        $xy \alphaequiv x \left(\makecycle{b,c}\cdot y\right)$ are $\alpha$-equivalences.
        Additionally, it is easy to verify that both
        $x \left(\makecycle{a,c}\cdot u\right) = \alpha_1 \cdots \alpha_n$ and
        $x \left(\makecycle{b,c}\cdot y\right) = \beta_1 \cdots \beta_n$ coincide up to the
        index $j$ where $\gamma_j = \ell + 1$. Thus, we obtain the decomposition
        $x' = \alpha_1 \cdots \alpha_{j - 1}$, $u' = \alpha_j \cdots \alpha_n$ and
        $y' = \beta_j \cdots \beta_n$. Because of the construction of the De Bruijn normal
        form, we have $\abs{x'} > \abs{x}$. Thus, we arrive at a sequence of $\alpha$-equivalent
        strings proving the equivalence $w \alphaequiv v$.

\mypar{Proof of \Cref{lem:aeinfbar}}
    \begin{proof}
        The \enquote{only if}-direction is trivial by definition of $\alpha$-equivalence on infinite bar strings.
        For the \enquote{if}-direction, we assume $p_0 \alphaequiv p_1$ and show $(u_0v_0^\omega)_n \alphaequiv
        (u_1v_1^\omega)_n$ for every $n \in \Nat$. Let $m$ be the length of $p_0$ (and $p_1$), then
        $(u_0v_0^\omega)_n \alphaequiv (u_1v_1^\omega)_n$ holds for every $n \leqslant m$. For $n > m$, we show
        equality of the De Bruijn normal forms $\nf((u_0v_0^\omega)_n) = \alpha_1 \cdots \alpha_n$ and
        $\nf((u_1v_1^\omega)_n) = \beta_1 \cdots \beta_n$ (\Cref{lem:aeDBNF}). We have equality
        of $\alpha_i$ and $\beta_i$ for $i \leqslant m$ by assumption. For $i > m$, suppose $\alpha_i \neq \beta_i$.
        We proceed by case distinction.
        \begin{enumerate}
        \item Suppose that $\alpha_i, \beta_i \in \names$. This means that there are free letters
        $a \neq b$ at position $i$ in $(u_0v_0^\omega)_n$ and $(u_1v_1^\omega)_n$ occurring by design in $v_0$ or $v_1$
        at some positions $i_0$ and $i_1$. However, note that due to the length of each $p_i$, every pair of letters that occur in some index of the infinite words already occur somewhere in the finite prefixes $p_0$ and $p_1$ at the same relative positions.
        In other words, there must be a position $j < m$ in $(u_0v_0^\omega)_n$ and $(u_1v_1^\omega)_n$ with the
        letters of position $i_0$ or $i_1$ in $v_0$ or $v_1$, 
        respectively. Since both letters are free at position $i$, they must be free at position $j$, thereby
        contradicting the $\alpha$-equivalence of $p_0$ and $p_1$.

      \item Suppose that $\alpha_i \in \names$ and $\beta_i \in \Nat$. This also contradicts our assumption of $p_0
        \alphaequiv p_1$: This means that at position $i$ there is a free letter $a \in \names$ in $(u_0v_0^\omega)_n$
        and a bar name or bound letter in $(u_1v_1^\omega)_n$. Suppose  that they occur at positions $i_0$ and $i_1$ in $v_0$  and $v_1$, respectively. Again, due to the length of $p_0$ and $p_1$, there must be a position $\abs{u_0} +  \max\set{\abs{v_0}, \abs{v_1}} < j < m$, where positions $i_0$ and $i_1$ in $v_0$ and $v_1$ occur.
        Clearly the letter $a$ at position $j$ in the first prefix is again free, leading to direct contradiction, if the letter at position $i_1$ in $v_1$ is a bar name.
        If the letter at position $i_1$ in $v_1$ is plain (and bound), then the corresponding bar name for position $i$ occurs either in $u_1$, the directly preceeding iteration of $v_1$ (whence in $v_1$ after position $i_1$), or in $v_1$ before position $i_1$.
        Since the position $j$ is chosen such that there is copy of $v_1$ preceeding it, this makes the letter at position~$j$ also plain and bound, contradicting the presumed $\alpha$-equivalence of $p_0$ and $p_1$.
        The case $\alpha_i \in \Nat$ and $\beta_i \in \names$ follows analogously.

      \item Suppose that $\alpha_i \neq \beta_i \in \Nat$.
        The case where there is a bar name in one string and a plain name in the other at position $i$, follows
        analogously from the previous case. 
        Therefore, suppose that there are two bar names at position $i$ in $(u_0v_0^\omega)_n$ and $(u_1v_1^\omega)_n$.
        Since the indices of bar names indicate the number of other bar names before that position in the bar string,
        there must be a position $j < i$, where there is a bar name in one string and a plain name in the other,
        leading to a contradiction as shown in the previous case.
        Lastly, if two plain names occur at position $i$ in $(u_0v_0^\omega)_n$ and $(u_1v_1^\omega)_n$ (and by design
        in $v_j$ at position $i_j$ for $j = 1, 2$), we see that either there is a position $j < m$ (repeating the
        positions $i_0$ and $i_1$) with $\alpha_j \neq \beta_j \in \Nat$ leading directly to a contradiction or
        there is some position in between where there is a bar name in one and plain name in the other, leading to
        a contradiction as shown in the previous case. 
      \end{enumerate}
      Therefore, we have $\alpha_i = \beta_i$. By~\Cref{lem:aeDBNF}, we conclude that $(u_0v_0^\omega)_n \alphaequiv (u_1v_1^\omega)_n$ holds, as desired.
    \end{proof}

  \begin{lemma}\label{lem:TnfstPP}
        De Bruijn normal forms for bar $\Sigma$-trees are stable under permutations: $\pi \cdot \nf(t)
        = \nf(\pi \cdot t)$ for $\pi \in \Perm(\names)$ and $t \in \tree_\names(\Sigma)$.
    \end{lemma}
    \begin{proof}
        We show the desired equality node by node: Let $\gamma.f(t_1, \dots, t_n)$ ($\gamma \in \names + \Nat$,
        $\nicefrac{f}{n} \in \Sigma$) be a node in $\nf(\pi \cdot t)$ corresponding to nodes
        $\pi(\alpha).f(\pi\cdot t_1', \dots, \pi \cdot t_n')$ ($\alpha \in \barNames$) in $\pi \cdot t$,
        $\alpha.f(t_1',\dots,t_n')$ in $t$ and $\delta.f(t_1'',\dots,t_n'')$ ($\delta \in \names + \Nat$) in $\nf(t)$.
        We show $\gamma = \pi \cdot \delta$ by case distinction:

        If $\delta \in \names$, then $\alpha$ is a free occurrence in $t$, thereby making
        $\pi(\alpha)$ a free occurrence in $\pi \cdot t$, which implies that $\gamma = \pi(\delta)$.

        If $\delta = k$ and there is no prior node $k.g(s_1, \dots, s_m)$ for
        $\nicefrac{g}{m} \in \Sigma$ in the branch from the root of $\nf(t)$, then
        $\alpha = \newletter{a} \in \barNames$ and there are $k - 1$ bar names before the node
        in the branch from the root of $t$. Since the same number of bar names also occur in
        $\pi \cdot t$, we have $\gamma = k$.

        Lastly, we consider the case where $\delta = k$ and there is a prior node
        $k.g(s_1, \dots, s_m)$ for $\nicefrac{g}{m} \in \Sigma$ in the brach from the
        root of $\nf(t)$. Let $k.g(s_1, \dots, s_m)$ be the earliest of those nodes. The
        corresponding node $\beta.g(s_1', \dots, s_m')$ in $t$ has
        $\beta = \newletter{a} \in \barNames$, making $\alpha = a$. In $\pi \cdot t$, the
        corresponding node is $\newletter{\pi(a)}.g(s_1'', \dots, s_m'')$, and this is the last
        node with $\newletter{\pi(a)}$ in the branch from the root to $\pi(\alpha).f$.  As in
        the argument for the previous case, the node corresponding to $k.g$ in
        $\nf(\pi\cdot t)$ is also $k.g$, hence $\gamma = k$.
    \end{proof}
    \begin{rem}\label{R:pi.t}
      Similarly to \Cref{R:pi.w}, we see that for a permutation $\pi$ that does not change any
      free letters of the tree $t$, we have $\nf(t) = \nf(\pi\cdot t)$.
    \end{rem}
    \begin{lemma}\label{lem:TnfstCC}
        Let $C$ be a context with $k$ holes and $s_1, \dots, s_k$ as well as $t_1, \dots, t_k$ be bar $\Sigma$-trees
        satisfying $\nf(s_i) = \nf(t_i)$ for $1 \leqslant i \leqslant k$. Then, $\nf(C[s_1, \dots, s_k]) =
        \nf(C[t_1, \dots, t_k])$.
    \end{lemma}
    \begin{proof}
      We show the desired equality of normal forms node by node: Let $\ell_i$ for
      $1 \leqslant i \leqslant k$ denote the number of bar names in $C$ in the branch from the
      root to the $i$-th hole. It is evident that nodes $\gamma.f(u_1, \dots, u_m)$, for
      $\gamma \in \names + \Nat$, $\nicefrac{f}{m} \in \Sigma$, corresponding to nodes in the
      context are equal in both normal forms.  Let $\gamma.f(u_1, \dots, u_m)$ for
      $\gamma \in \names + \Nat$, $\nicefrac{f}{m} \in \Sigma$ in $\nf(C[s_1, \dots, s_k])$ be
      a node corresponding to $\alpha.f(u_1', \dots, u_m')$ in $s_i$ for
      $1 \leqslant i \leqslant k$ and to $\delta.f(u_1'', \dots, u_m'')$ in $\nf(s_i) =
      \nf(t_i)$. We shall now prove that $\gamma$ only depends on $\nf(s_i)$ and $C$.  
      We proceed by case distinction:
        
      Firstly, if $\delta \in \Nat$, then we show $\gamma = \delta + \ell_i$. Indeed, if
      $\alpha = \newletter{a} \in \barNames$, then there are $\delta - 1$ bar names in $s_i$ on
      the branch from the root to the node $\alpha.f$. Thus, we have $\delta - 1 + \ell_i$ bar
      names in $C[s_1, \dots, s_k]$ on the branch from the root to this node $\alpha.f$. Hence,
      by definition $\gamma = \delta + \ell_i$. If $\alpha$ is a plain name bound by a bar name
      in $s_i$ in the branch from the root to the node, then that bar name is replaced with the number
      $\delta + \ell_i$ in $\nf(C[s_i, \ldots, s_k])$; in symbols: $\gamma = \delta + \ell_i$.

      Secondly, if $\delta \in \names$ and there is no node $\newletter{\delta}.g$ in the
      branch from the root to hole $i$ in the context $C$, then $\delta$ remains free in the
      branch from the root to the node in $C[s_1, \ldots, s_k]$. Thus, we have
      $\gamma = \delta$.

      Lastly, if $\delta \in \names$ and there is a node $\newletter{\delta}.g$ in the branch
      from the root to hole $i$ in the context $C$, then take the last of these. In fact,
      this (last) node corresponds to $j.g$ in $\nf(C[s_1, \dots, s_k])$. Thus, we have $\gamma
      = j$.

      This finishes our argument that $\gamma$ only depends on $\nf(s_i)$ and $C$. Since
      $\nf(s_i) = \nf(t_i)$ for all $i = 1, \ldots, k$, we can thus conclude that
      $\nf(C[s_1, \ldots, s_k])$ and $\nf(C[t_1, \ldots, t_k])$, as desired.
    \end{proof}

\mypar{Proof of \Cref{lem:aeTrees}}
      \noindent $(\Rightarrow)$: Let $s \alphaequiv t$ in one step at the root, that is,
      $s = \newletter{a}.f(s_1, \dots, s_n)$, $\newletter{b}.f(t_1, \dots, t_n) = t$ for
      $a, b \in \names$, $\nicefrac{f}{n} \in \Sigma$ and $s_i, t_i \in \tree_\names(\Sigma)$
      with $\braket{a}(s_1, \dots, s_n) = \braket{b}(t_1, \dots, t_n)$. This is without loss of
      generality, since contexts do not change the equality of normal forms
      (\Cref{lem:TnfstCC}). Therefore, we have a fresh $c \in \names$ such that
      $\makecycle{a,c} \cdot s_i = \makecycle{b,c} \cdot t_i$ for all
      $1 \leqslant i \leqslant n$ and thus also
      $\makecycle{a,c} \cdot s = \makecycle{b,c} \cdot t$. By~\Cref{lem:TnfstPP}
      and~\Cref{R:pi.t}, we see that
      \[
        \nf(s) = \nf(\makecycle{a,c} \cdot s) = \nf(\makecycle{b,c} \cdot t) = \nf(t).
      \]

      \noindent $(\Leftarrow)$: Let $\nf(s) = \nf(t)$. We show $s \alphaequiv t$ using an
      iterative process as follows. If $s = t$, then we are done. Otherwise, let $C$ be the
      maximal context such that $s = C[s_1, \dots, s_k]$ and $t = C[t_1, \dots, t_k]$. By the
      maximality of $C$ and since $\nf(s) = \nf(t)$, we know that there exist $s_j$ and $t_j$
      such that $s_j = \newletter{a}.f(s_1',\dots, s_n')$ and
      $t_j = \newletter{b}.f(t_1',\dots, t_n')$ for $a\neq b \in \names$ and
      $\nicefrac{f}{n} \in \Sigma$. Now choose a $c \in \names$ fresh for both $s_j$ and
      $t_j$. Then, we have $s_j \alphaequiv \makecycle{a,c} \cdot s_j$ and
      $t_j \alphaequiv \makecycle{b,c} \cdot t_j$ and, moreover, $\makecycle{a,c} \cdot s_j$
      and $\makecycle{b,c} \cdot t_j$ coincide in the root node $\newletter{c}.f$.  Thus, we
      have a non-trivial context $C'$, say with $i$ holes, such that
      $s_j \alphaequiv C'[s_1^j, \dots, s_i^j]$ and $t_j \alphaequiv C'[t_1^j, \dots,
      t_i^j]$. Therefore we have a non-trivial context $C''$ which is strictly bigger than $C$
      and such that
      \[
        s \alphaequiv C''[s_1, \dots, s_{j-1}, s_1^j, \dots, s_i^j, s_{j+1}\dots, s_k]
        \ \ \text{and}\ \
        t \alphaequiv C''[t_1, \dots, t_{j-1}, t_1^j, \dots, t_i^j, t_{j+1}\dots, t_k].
      \]
      If the two right-hand sides are equal, then we are done. Otherwise, we continue to pick 
      Now either the two right-hand sides are equal or we continue as before. By the finiteness
      of $s$ and $t$, this eventually yields $s \alphaequiv t$.

      \section{Details for~\Cref{sec:learningRNNA}}
      \mypar{Proof of~\Cref{lem:correctEverything}}
      Clearly, if the learner $\learn_{\textsf{bar}}$ terminates, then it has inferred a correct bar automaton for the unknown bar language $L_\teach$. Thus, we only need to establish $\learn_{\textsf{bar}}$'s query complexity. As indicated above, the key observation is that all answers the internal learner \learn receives from the \tass correspond to answers of a teacher for the regular language $L_\teach\restriction \barNames_0$; that is,
\begin{enumerate}
\item\label{claim-1} If \learn asks a membership query $w\in
  \barAs_0/\barAw_0/\bartree[\barNamess_0](\Sigma)$, 
  the \tass answers `Yes' iff $w\in L_\teach\restriction \barNames_0$.
\item\label{claim-2} If \learn asks an equivalence query with hypothesis $\H$ (a bar automaton
  over $\barNames_0$), then the answer of the \tass (if any) is an element of the symmetric difference $L(\H) \oplus (L_\teach\restriction \barNames_0)$. 
\end{enumerate}
Note that the \tass might not answer an equivalence query by \learn at all: if $\H$ satisfies
$L_\alpha(\H)=L_\teach$, then $\learn_{\textsf{bar}}$ successfully terminates without running $\learn$
to completion. The claimed complexity bound for $\learn_{\textsf{bar}}$ is then immediate: Since
$\learn_{\textsf{bar}}$ simply forwards each of $\learn$'s membership and equivalence queries to
\teach, the total number of $\learn_{\textsf{bar}}$'s queries is at most $M(L_\teach\restriction
\barNames_0)$ and $E(L_\teach\restriction \barNames_0)$, respectively. It remains to prove the
statements in \labelcref{claim-1,claim-2}.

\smallskip\noindent
\emph{Proof of \cref{claim-1}.} Since a membership query $w$ by \learn is a bar input over $\barNames_0$, we have $w\in L_\teach\restriction \barNames_0$ iff $w\in L_\teach$.
Therefore \teach's answer to the query `$w\in L_\teach$?', which the \tass forwards to \learn, is also a correct answer for \learn's query `$w\in L_\teach\restriction \barNames_0$?'.      

\smallskip\noindent \emph{Proof of \cref{claim-2}}. Suppose that $\H$ is an incorrect hypothesis for $L_\teach$, that is, $L_\alpha(\H)\neq L_\teach$, so that \teach returns an element $w_\teach$ of the symmetric difference $L_\alpha(\H) \oplus L_\teach$ to $\learn_{\textsf{bar}}$.
Let $w\alphaequiv w_\teach$ be the $\alpha$-equivalent bar input over $\barNames_0$ chosen by the~\tass in step $\mathsf{1}$. We consider two cases:
\begin{itemize}
\item Case 1: $w_\teach\in L_\teach\setminus L_\alpha(\H)$. Then $w\in L_\teach\setminus L_\alpha(\H)$ because both $L_\teach$ and $L_\alpha(\H)$ and hence $L_\teach\setminus L_\alpha(\H)$ are closed under $\alpha$-equivalence. Since $w\not\in L_\alpha(\H)$, there exists no $w'\in L(\H)$ such that $w'\alphaequiv w$, so the \tass returns $w$ to \learn after step $\mathsf{2}$. Then $w\in (L_\teach\restriction \barNames_0)\setminus L(\H)$, in particular $w\in  L(\H)\oplus (L_\teach\restriction \barNames_0)$ as claimed.
\item Case 2: $w_\teach\in L_\alpha(\H)\setminus L_\teach$. Then $w\in L_\alpha(\H)\setminus L_\teach$, analogous to Case 1. Since $w\in L_\alpha(\H)$, there exists $w'\in L(\H)$ such that $w'\alphaequiv w$. The \tass picks one such $w'$ in step $\mathsf{2}$ and returns it to \learn. Note that $w'\not\in L_\teach$ because $w\not\in L_\teach$ and $L_\teach$ is closed under $\alpha$-equivalence. Therefore $w'\in L(\H)\setminus (L_\teach\restriction \barNames_0)$, whence $w'\in L(\H)\oplus (L_\teach\restriction \barNames_0)$ as claimed.
\end{itemize}   

      \mypar{Proof of~\Cref{prop:complTAS1}}
      We describe an algorithm to find a bar input $w$ over $\barNames_0$ that is $\alpha$-equivalent to the given counterexample $w_\teach$.

      For bar word and tree automata, turn $w_\teach$ into its De Bruijn normal form and then replace all indices by letters from $\barNames_0$ to get an $\alpha$-equivalent bar input over $\barNames_0$. Suitable replacements can be found by using a list of \emph{constraints},
      that is, for each index in the normal form a list of letters and indices to which the corresponding name must not be equal.
      The constraints are computed incrementally from the end of the bar input in polynomial time.

      For bar Büchi automata, first build a bar automaton over $\barNames_1=\setw{\alpha\in \barNames}{\text{$\alpha$ occurs in $w_\teach$}}$
      accepting just the counterexample $w_\teach=uv^\omega$ with $\mathcal{O}(\card{uv})$ states; note that $\card{\barNames_1} \leq \card{uv}$.
      Then construct the closure of that automaton w.r.t.\ the alphabet $\barNames_0\cup\barNames_1$. It accepts all $w\in (\barNames_0 \cup \barNames_1)^\omega$ such that $w \alphaequiv w_\teach$ and
      needs $\mathcal{O}(\card{uv} \cdot 2^{(\card{\barNamess_0}+\card{\barNamess_1})\cdot (\log(\card{\barNamess_0}+\card{\barNamess_1}) + 1)})$ states (\Cref{prop:barautclosure}).
      Next, restrict the input alphabet of the closure to $\barNames_0$ by removing all $\barNames_1\setminus \barNames_0$-transitions. 
      The resulting bar automaton $\A$ accepts the language $L(\A)=\setw{w \in \barAw_0}{w \alphaequiv w_\teach}$. Thus, it remains to look for some
      bar input $w$ in $L(\A)$. This, however, is standard: Perform a depth-first search to detect a cycle in $\A$ (with label $v_1$) that contains a final state
      and is reachable from the initial state (via some path with label $u_1$). Then $w=u_1v_1^\omega$ is an ultimately periodic bar string contained in $L(\A)$.
      Overall, this algorithm uses space linear in $\card{uv}$ and exponential in $\card{\barNames_0}$ and $\card{\barNames_1}$. 
      \mypar{Proof of~\Cref{prop:complTA}}
      We describe an algorithm that computes, for a given bar input $w$ over $\barNames_0$ and hypothesis $\H$, an $\alpha$-equivalent $w'\alphaequiv w$ accepted by $\H$ (if it exists). 
      Take a bar automaton~$\A$ for the closure of $\{w\}$ with
      $\mathcal{O}(\card{w} \cdot 2^{\card{\barNamess_0}\cdot (\log(\card{\barNamess_0}) + 1)})$ states (\Cref{prop:barautclosure}); thus $L(\A)=\setw{w'}{w' \alphaequiv w}$. Then search for a bar input contained in the intersection
      $L(\A) \cap L(\H)$, where a bar automaton for the intersection is formed via a standard product construction. 
      This requires space linear in $\card{w}$ and the number of states of $\H$, and exponential in 
      $\card{\barNames_0}$.
\end{document}